%% file: ms.tex
\documentclass[journal,onecolumn,transmag]{IEEEtran}
\usepackage{amsmath}
\usepackage{url}
\usepackage{latexsym}
\usepackage{amssymb}
\usepackage{fancyhdr}
\usepackage{amsmath}
\usepackage{color}
\usepackage{scrextend}
\usepackage{graphicx}
\usepackage{hyperref}
\usepackage{caption} 
\usepackage{subcaption}
\usepackage{amsthm}
\usepackage{float}
\usepackage[nottoc]{tocbibind} 
\usepackage{cleveref}
\usepackage{lipsum}
\usepackage[ruled,vlined]{algorithm2e}
\newtheoremstyle{mytheorem}
  {0pt}
  {0pt}
  {\itshape}
  {}
  {\bfseries}
  {:}
  {0.5em}
  {}
\theoremstyle{mytheorem}
\newtheorem{theorem}{Theorem}
\newtheorem{lemma}{Lemma}

\newtheorem{corollary}{Corollary}

\newtheorem{claim}{Claim}
\crefname{claim}{claim}{claims}
\newtheorem{definition}{Definition}

\newtheorem{remark}{Remark}
\newtheorem{example}{Example}
\newcommand{\deEp}{{\delta_\epsilon}}

\newcommand{\Z}{\mathbb{Z}}
\newcommand{\R}{\mathbb{R}}
\newcommand{\N}{\mathbb{N}}

\newcommand{\bw}{\mathbf{w}}
\newcommand{\bu}{\mathbf{u}}

\newcommand{\bh}{\mathbf{h}}
\newcommand{\by}{\mathbf{y}}
\newcommand{\bz}{\mathbf{z}}
\newcommand{\bH}{\mathbf{H}}
\newcommand{\bB}{\mathbf{B}}

\newcommand{\bA}{\mathbf{A}}

\newcommand{\bI}{\mathbf{I}}
\newcommand{\bS}{\mathbf{S}}

\newcommand{\bD}{\mathbf{D}}

\newcommand{\bU}{\mathbf{U}}

\newcommand{\bSig}{\mathbf{\Sigma}}

\newcommand{\bR}{\mathbf{R}}

\newcommand{\bx}{\mathbf{x}}
\newcommand{\bM}{\mathbf{M}}

\newcommand{\ba}{\mathbf{a}}
\newcommand{\bb}{\mathbf{b}}
\newcommand{\be}{\mathbf{e}}

\newcommand{\normSqr[1]}{\ensuremath{\|#1\|^2}}

\newcommand{\p[1]}{\Pr\left\{ #1  \right\}}

\newcommand{\defined}{\stackrel{\triangle}{=}}

\newcommand{\off}[1]{}

\crefname{algocf}{alg.}{algs.}
\Crefname{algocf}{Algorithm}{Algorithms}

\begin{document}
\input{tex_files/0_header_and_titlepage/abstract_and_title.tex}

\section{Introduction}\label{sec:intro}	
\input{tex_files/intro/intro.tex}

\section{Bounds on the Effective SNR of a Single Vector of Coefficients}\label{sec:bounds_on_SNR_eff}
\input{tex_files/bounds_on_eff_snr/lower_bound_eff_snr_intro.tex}	
\input{tex_files/bounds_on_eff_snr/bounds_on_eff_snr_single_equation.tex}

\section{Practical IF-Based Schemes}\label{sec:practical_sub_opt_IF}	
\input{tex_files/practical_IF_based_schemes/practical_IF_based_schemes_intro.tex}
\input{tex_files/practical_IF_based_schemes/B_IF.tex}

\input{tex_files/practical_IF_based_schemes/BB_IF.tex}
\input{tex_files/practical_IF_based_schemes/results.tex}

\section{Appendix}\label{sec:appendix}
\input{tex_files/appendix/appendix_LB_construction.tex}
\input{tex_files/appendix/Practical_IF_Based_Schemes.tex}

\input{tex_files/appendix/proof_of_the_lb_theorem.tex}

\bibliographystyle{IEEEbib}
\bibliography{My_bib}		
\end{document}

%% file: tex_files/0_header_and_titlepage/abstract_and_title.tex
\title{Integer Forcing: Effective SNR Distribution and Practical Block-Based Schemes}
\author{%
  \IEEEauthorblockN{Ron Meiry, Omer Gurewitz and Asaf Cohen}\\
  \IEEEauthorblockA{%
    Department of Communication Systems Engineering\\
Ben-Gurion University of the Negev\\
    Email: \{ronmeir@post., gurewitz@, coasaf@\}bgu.ac.il}
  }
\maketitle
\begin{abstract}
Integer Forcing (IF) is a novel linear receiver architecture, where instead of separating the codewords sent by each transmitter, and decoding them individually, forces integer-valued linear combinations at each receive antenna, and decodes the linear combinations. The original codewords are obtained by inverting the integer-valued matrix. While demonstrating superior performance, IF requires complex optimization in order to find the optimal linear combinations, and demands either using multi-level nested lattice codes, or reducing the rates of all transmitters to equal the weakest one. Finally, the distribution of the resulting effective SNR is hard to evaluate.

In this paper, we first give simple upper and lower bounds on the effective SNR of a single linear combination in IF. These expressions allow us to easily bound the distribution of the effective SNR for any given linear combination used. We then suggest two simple block-based IF schemes. These schemes, while sub-optimal, significantly reduce the complexity of the optimization process, and, more importantly, do not require reducing the rates of all transmitters, as decoding is done block-wise. Finally, we bound the distribution of the effective SNR of the decoding schemes, and show via simulations the superiority of block-wise schemes at low SNR.
\end{abstract}

%% file: tex_files/intro/intro.tex
Linear receivers, such as the \emph{Zero Forcing} (ZF) and the \emph{Minimum Mean Square Error} (MMSE) receivers, where developed to reduce the receiver complexity relatively to \emph{Maximum Likelihood} (ML), and are commonly used in \emph{Multiple Input Multiple Output} (MIMO) communication.  The \emph{Integer Forcing} (IF)  receiver \cite{zhan2014integer} is a more recent technique, which uses \emph{Compute and Forward} (CnF) \cite{nazer2011compute} as a building block and allows the receiver to decode linear combinations of the sent symbols. It can be considered as a two phased decoding process. First, choose a full ranked matrix representing linear combinations of the transmitted messages. Next, at each virtual antenna, decode the corresponding linear combination. Different linear combinations result in different rates at which the decoder is able to decode. The strength of IF lies in the decoder's ability of to find a suitable matrix. \Cref{fig:IF_scheme} depicts the receiver architecture. Note that setting $\bA=\bI$ and $\bB$ to be the pseudoinverse of $\bH$, ZF is achieved, and by setting $\bB= \left(\frac{1}{P} \bI + \bH^T\bH\right)^{-1}\bH^T$, MMSE is achieved. 
Since $\bA$ can be any full ranked matrix and not necessarily the identity matrix, IF can be superior to ZF or MMSE. In fact, IF achieves capacity up to a constant gap \cite{ordentlich2015precoded}.

However, along side with the fact that choosing the best linear combination for any channel realization is the strength of IF, it is also its main drawback. Computing the optimal integer-valued matrix $\bA$ is exponential in its size. Hence, an algorithm to choose a suboptimal $\bA$ which has good performance is required. In other words, in order to be able to use IF in practice, the number of linear combinations (represented as vectors) being checked must be limited. Moreover, IF requires either using nested lattice codes, to allow the receiver to decode linear combinations where each symbol is encoded at a different rate, or restricting all transmitting antennas to the lowest rate. Thus, it is either a complexity burden, or a hit on the sum-rate.


\subsection{Main Contribution}
The contribution of this work is twofold. In the first part, we present lower and upper bounds on the effective SNR of a linear combination in IF. Then, under a fixed linear combination, we calculate the distribution of the lower bound. The distribution of the effective SNR in IF is interesting by its own right, but also essential for the analysis of the schemes presented in the second part of this work.

The second part offers two distributed IF-based schemes: \emph{Block-IF} (B-IF) and \emph{Norm-Bounded-IF} (NB-IF). Both enforce $\bA$ to be a block matrix, thus, $\bA$ can be found efficiently. Moreover, when block matrices are used, the transmitted messages mix only with other messages from the same block, hence, each block can use a different rate. Furthermore, both schemes are lower and upper bounded and the distribution of the lower bound from the first part is used for the analysis.

Finally, using simulations, we show that for relatively small transmission powers, e.g., less then 3dB when $M_R$=$M_T$=4, both B-IF and NB-IF are not only more efficient than IF, but also achieve higher rates. In fact, our numerical results show that as long as the transmission power $P$ is smaller than 12dB, NB-IF is better than the upper bound on IF which was demonstrated by Ordentlich \emph{et al.} \cite{ordentlich2015precoded}.
\subsection{Related Works}
A large body of work is available on MIMO techniques. A survey can be found in \cite{mietzner2009multiple}. Herein, we only briefly mention the most relevant ones.
The optimal receiver (Joint ML) requires searching for the most likely set of transmitted streams. For the general case, this problem is very complicated.  In order to reduce the computation in exchange for rate, linear receivers such as ZF and MMSE were developed. The performance of ZF, including the distribution of the SNR, is well known. For MMSE, it is known only for special cases. Analysis of performance of ZF and MMSE can be found in \cite{jiang2011performance}.

\emph{Computer and Forward} (CnF) was presented in \cite{nazer2011compute} as a relaying method which allows relays to decode linear combinations of the codewords sent. CnF is based on a nested lattice codebook, which allows decoding linear combinations of the transmitted streams. After a linear combination decoded by a relay, it is forward to the next relay or to the destination.
\emph{Integer Forcing} (IF) was first presented in \cite{zhan2014integer} and builds on the concepts in CnF. It is a MIMO linear receiver which instead of decoding the original codewords, decodes linear combinations of them. 
In \cite{zhan2010integer}, a distributed architectures which reduces the complexity at the receiver for IF implementation is presented. This work also investigates successive interference cancellation (SIC) for IF.
Upper and lower bound on the rate of IF can be found in \cite{ordentlich2015precoded}. \cite{ordentlich2013successive} shows how to achieve the capacity using successive IF, when allowing different transmission rates. In addition, \cite{ordentlich2013successive} shows that using successive IF is more fair than MMSE-SIC, which also achieves the capacity, since it supports a variety of rate tuples, including those achievable in MMSE-SIC. 

At the heart of IF is the choice of the coefficients for the linear combinations to be decoded. \cite{mejri2013practical} presents practical methods to choose $\bA$ base of the LLL algorithm and by bounding the search to a sphere of a given radius. In \cite{richter2012efficient} a \emph{branch-and-bound} based algorithm for CnF is provided. The algorithm calculates the coefficient vector, which results in the highest computation rate at a single node.
In \cite{wei2013integer}, practical and efficient suboptimal algorithms to design the IF coefficient matrix are given. Those algorithms are based on the slowest descent method. In order to design the IF coefficient matrix with integer elements, first a feasible searching set is generated based on the slowest descent method. Then, integer vectors within the searching set are picked in order to construct the full rank IF coefficient matrix.

\section{Preliminaries}
\subsection{Model}
We consider an up-link MIMO Channel with $M_T$ transmitting antennas and a base station (BS) equipped with $M_R$ receiving antennas. We assume that $2\leq M_T \leq M_R$.  When $M_T$ users are transmitting simultaneously, each sub channel from a transmitting antenna to a receiving antenna is represented by a channel coefficient $h_{ij} \in \R$.
Throughout, we use boldface lowercase to refer vectors, e.g., $\bh_i = [h_{i1},h_{i2},\ldots,h_{i M_R}]^T \in \R^{M_R}$ represents the channel vector of the $i^{th}$ user to the BS. Boldface uppercase are used to refer to matrices, e.g., the channel matrix between the users to the BS is denoted by $\bH = [\bh_1,\bh_2,\ldots,\bh_{M_T}] \in \R^{M_R \times M_T}$. It is common to assume that the entries of $\bH$ are i.i.d.\ Normal RVs. We assume unit variance. We further assume that the receiver has channel state information, i.e., the matrix $\bH$ is known to the receiver. The signal transmitted by the $i^{th}$ user is  $x_i \in \R$. At the receiving antennas, we assume i.i.d.\ Normal noise with variance $N_0$, which is denoted by $\bz \in \R^{M_R}$. Accordingly, $\by \in \R^{M_R}$, the received vector at the BS is $\by=\bH \bx+\bz$,
where $\bx = [x_1, \ldots, x_{M_T}]^T$ is the transmitted vector. The average transmission power of each user is bounded by $P$. 
\subsection{Integer-Forcing Linear Receiver}\label{sc:IF_intro}
\begin{figure}
	\centering
		\includegraphics[width=0.7\textwidth]{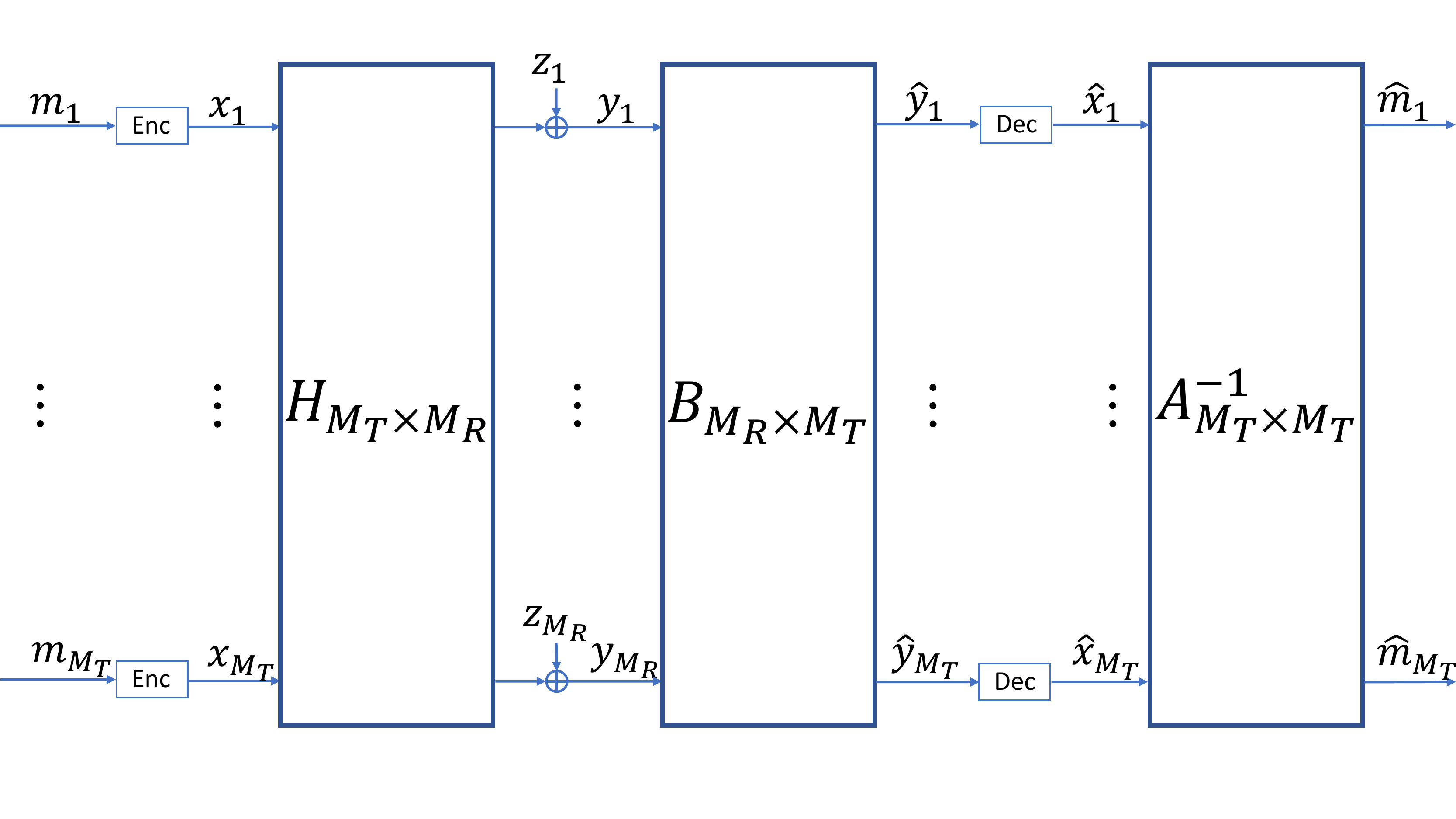}
	\caption{IF scheme block diagram} 
	\label{fig:IF_scheme}
\end{figure}
The IF linear receiver decodes linear combinations of codewords in each of its receiving antennas. The coefficients of the combinations are represented as elements of a matrix $\bA$, and are chosen by the receiver, based on the channel matrix. Thus, CSI is required at the receiver. Decoding the linear combinations is based on CnF, a building block of IF. Specifically, the codewords are taken from a lattice codebook \cite{zamir2014lattice}, hence each linear combination is a codeword by itself, and can be decoded (if the noise is below a certain level). After the recovery process of each linear combination, the originally transmitted vector can be found iff all the vectors of coefficients at the receiver form a full rank matrix $\bA^{M_T\times M_T}$. Figure \ref{fig:IF_scheme} depicts the scheme: the matrix $\bB$ is used to approximate $\bA$ at the decoder's input, then after decoding each stream, $\bA^{-1}$ is used. 

Let $\ba_k\in\Z^{M_T}$ denote the integer vector (of coefficients) of the $k^{th}$ linear combination. The effective SNR of IF is \cite{zhan2014integer}
\begin{equation}\label{eq:IF_eq_snr_eff_a_and_b}
SNR_{eff}(\bH,\ba_k,\bb_k) =  \frac{P}{\normSqr[\bb_k]+P\normSqr[\bH^T\bb_k-\ba_k]}.
\end{equation} 
\Cref{eq:IF_eq_snr_eff_a_and_b} gives insight on choosing $\ba_k$ and $\bb_k$.  For example, if $P$ is large, $\ba_k$ and $\bb_k$ should be chosen s.t. $\normSqr[\bH^T\bb_k-\ba_k]$ is small, since this expression is multiplied by $P$. However, if $P$ is small, the important summand is $\normSqr[\bb_k]$. In \cite[Section F]{zhan2014integer} it is shown that 
\begin{equation}
\forall k: \|\ba_k\| < 1+ \sqrt{P}\lambda_{max}(\bH),
\label{eq:bound_the_norm_of_a}
\end{equation}
where $\lambda_{max}(\bH)$ is maximal eigenvalue is $\bH^T\bH$. Otherwise, the rate of IF is zero.

Given $\ba_k$ and $\bH$, $\bb_k$ can be computed. Hence, the effective SNR, and the corresponding rate achieved by IF using a coefficient vector $\ba_k$, are given by \cite{ordentlich2015precoded}
\begin{equation}
SNR_{eff}(\bH,\ba_k) = \left({\ba_k^T\left(\bI+P\bH^T\bH\right)^{-1}\ba_k}\right)^{-1} \label{eq:IF_eq_snr_eff}
\end{equation} 
and 
\begin{equation}
R_{IF}(\bH,\ba_k) = \frac{1}{2}\log_2 SNR_{eff}(\bH,\ba_k).
\end{equation} 
In order to decode the original messages, $M_T$ linearly independent vectors are needed. Those vectors can represented as a matrix $\bA = [\ba_1,\ldots,\ba_{M_T}] \in \Z^{M_T\times M_T}$, s.t. $rank(\bA)=M_T$. 

Basic IF is designed such that all the messages are taken from the same codebook and can be linearly combined using integer coefficients. As a result, all the transmitters are forced to use the same codebook, hence the same transmission rate. Thus, the effective SNR and the corresponding sum rate of IF given $\bA$ are:
\begin{equation*}
SNR_{eff}(\bH,\bA) \defined \min_{\ba \in \bA}SNR_{eff}(\bH,\ba)
\end{equation*}
and 
\begin{equation*}
R_{IF}(\bH,\bA) \defined M_T \cdot \min_{\ba \in \bA} R(\bH,\ba).
\end{equation*}
This means that the optimal IF effective SNR and sum rate, are given by 
\begin{equation*}
SNR_{eff}(\bH) \defined \max_{\bA \in \Z^{M_T\times M_T}} SNR_{eff}(\bH,\bA),
\end{equation*}
 and 
\begin{equation*}
R_{IF}(\bH) \defined \max_{\bA \in \Z^{M_T\times M_T}} R(\bH,\bA), 
\end{equation*}
where $rank(\bA)=M_T$. An upper bound on the effective SNR of IF is presented in \cite{ordentlich2015precoded}: 
\begin{equation*}
SNR_{eff}(\bH) \leq \min_{\ba \in \Z^{M_T}\setminus 0} \ba^T(\bI+P\bH^T\bH)\ba.
\end{equation*}
In fact, using different rates at each transmitter is possible, yet requires much more complicated methods. In \cite{nazer2011compute}, a demonstration and an achievable region for decoding two different vectors using successive cancellation is presented. This idea can be generalized to more than two vectors. In \cite{ordentlich2013successive} the Successive Integer-Forcing (S-IF) method, which is based on noise cancellation, is presented.

%% file: tex_files/bounds_on_eff_snr/lower_bound_eff_snr_intro.tex
In this section, we present lower and upper bounds on the effective SNR for a single vector. Using the bounds, we are able to characterize the \emph{distribution of the effective SNR} for a fixed vector. This characterization gives insight on the achievable performance in general, and is useful for the analysis of the schemes presented in \Cref{sec:practical_sub_opt_IF}.

%% file: tex_files/bounds_on_eff_snr/bounds_on_eff_snr_single_equation.tex
Our first result is a lower bound.
\begin{theorem}\label{lem:LB_lower_bounds_SNR_eff}
Let $\ba \in \Z^{M_T}$ s.t.\ $\normSqr[\ba]>0$. Then,
\begin{equation}
SNR_{eff}(\bH,\ba) \geq {\left({  \ba^T\left({P\bH^T\bH}\right)^{-1}\ba }\right)^{-1}}.
\end{equation}
\end{theorem}
The intuition behind \Cref{lem:LB_lower_bounds_SNR_eff} comes from the fact that  $\bH^T\bH$ is almost surely a positive definite matrix. Hence, $\frac{\left({\ba^T \left(\bI+P\bH^T\bH\right)^{-1}\ba}\right)^{-1}}{\left({  \ba^T\left({P\bH^T\bH}\right)^{-1}\ba }\right)^{-1}}$
 is a ratio between quadratic forms of the inverses of positive definite matrices. Intuitively, $\bI$, which appears in the effective SNR, enlarges the singular values of the corresponding quadratic form. The complete proof can be found in \Cref{subsec:construction_of_LB}.

\begin{definition}\label{def:LB_def}
$LB(\bH,\ba) \defined \left({  \ba^T(P\bH^T\bH)^{-1}\ba }\right)^{-1} = \frac{P}{\ba^T\bS^{-1}\ba}$,
\end{definition}
 where $\bS = \bH^T\bH$. When $\ba$ is fixed, the distribution of this lower bound can be computed. 
\begin{corollary}\label{corl:LB_distribution}
For a fixed $\ba$, the lower bound on the effective SNR has the following gamma distribution
\begin{align}
LB(\bH,\ba) \sim \Gamma\left(\frac{M_R-M_T+1}{2},\frac{2P}{\normSqr[\ba]}\right). \label{eq:LB_distribution}
\end{align}
\end{corollary}
The proof is a consequence of \cite[proposition 8.9]{eaton}, combined with the connection between the $\chi^2$ and the $\Gamma$ distributions.

Note that due to the $\frac{1}{\normSqr[\ba]}$ factor, without CSI the best vector of coefficients for the lower bound is a unit vector $\be_i$. Hence, the best choice without CSI is $\bA = \bI$. In that case, the lower bound is distributed exactly like ZF \cite[Lemma 1]{li2006distribution}, and IF becomes MMSE. Note also that the dimensions of $\bH$ do not effect the distribution, only the difference between the number of rows and the number of columns.

We now turn to a corresponding upper bound.
\begin{theorem}\label{trm:UB_eff_snr}
Let $\ba\in\Z^{M_T}$ s.t $\normSqr[\ba]>0$. Then, 
\begin{equation}
SNR_{eff}(\bH,\ba) \leq P\left({\ba^T{\bS}^{-1}\ba }\right)^{-1} \sum_{k=0}^\infty \left(\frac{tr(\bS^{-1})}{P}\right)^k . 
\end{equation}
\end{theorem}
\begin{proof}[Proof of \Cref{trm:UB_eff_snr}]
The proof is directly from \Cref{clm:LB_ratio_IF_vs_my_known_dist} by setting $\epsilon = \frac{1}{P}.$
\end{proof}
Note that the lower bound times an infinite series is exactly the upper bound. Thus, setting 
\begin{equation}
C(\bH,P) \defined \sum_{k=0}^\infty \left(\frac{tr(\bS^{-1})}{P}\right)^k,
\label{eq:def_of_C}
\end{equation}
 we have 
\begin{align*}
P\left({\ba^T{\bS}^{-1}\ba }\right)^{-1} \sum_{k=0}^\infty \left(\frac{tr(\bS^{-1})}{P}\right)^k=LB(\bH,\ba) \cdot C(\bH,P)
\end{align*}
and hence
\begin{align}\label{eq:LB_and_UB_to_SNR_eff}
LB(\bH,\ba) \leq SNR_{eff}(\bH,\ba)
& \leq C(\bH,P)\cdot LB(\bH,\ba).
\end{align}

On one hand, there always exists a non-zero probability for the series above not to converge. In that case, the upper bound is meaningless. On the other hand, it is not hard to see that when $tr(\bS^{-1})<P$, $C(\bH,P) =\frac{1}{1-\frac{tr(\bS^{-1})}{P}}$.
Hence, 
\begin{equation}
1\leq \frac{SNR_{eff}(\bH,\ba)}{LB(\bH,\ba)} \leq \frac{1}{1-\frac{tr(\bS^{-1})}{P}}.
\label{eq:LB_and_UB_eq_r_tight_for_high_P}
\end{equation}
\begin{remark}\label{rm:UB_is_tight_when_P_goes_infinity}
Note that when $P\rightarrow \infty$, $\frac{tr(\bS^{-1})}{P} \longrightarrow 0$. Thus, when $P\rightarrow\infty$, the lower and upper bounds are tight. \Cref{clm:find_P_st_C_smaller_then_t_WHP}, which can be found in \Cref{subsec:construction_of_LB}, can be used bound the probability that $C(\bH,P)$ does not converge.
\end{remark}

\begin{figure}
    \centering
    \begin{subfigure}{0.47\textwidth}
        \includegraphics[width=\textwidth]{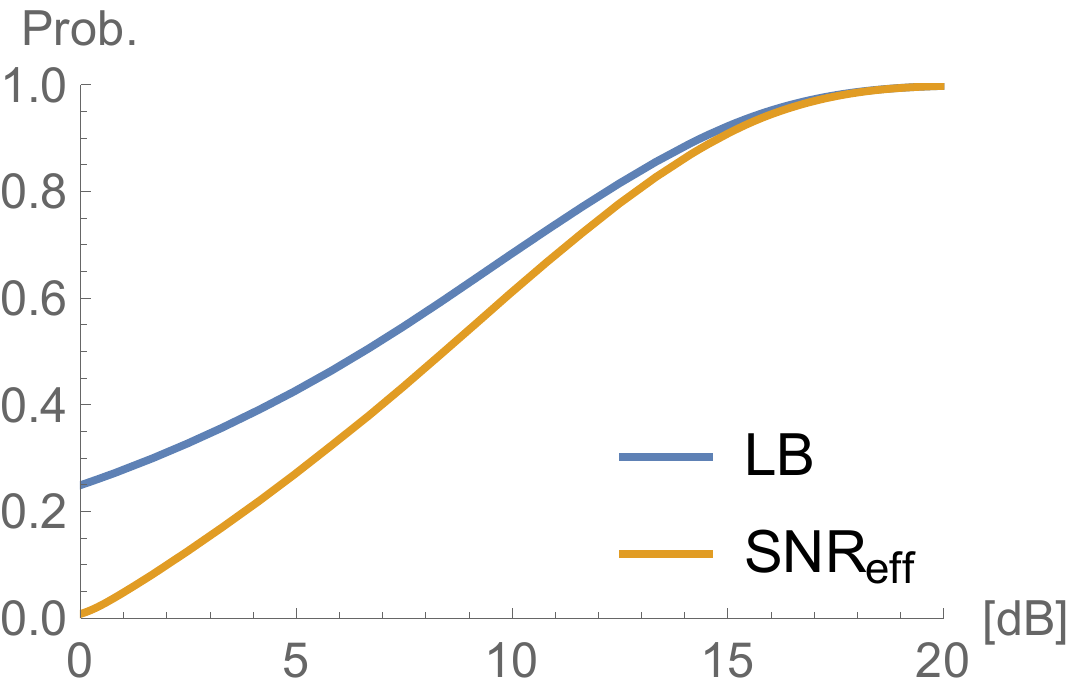}
        \caption{Lower bound on effective SNR.}
        \label{fig:LB_SNR_eff_Mt_2_Mr_2}
    \end{subfigure}
    \begin{subfigure}{0.47\textwidth}
        \includegraphics[width=\textwidth]{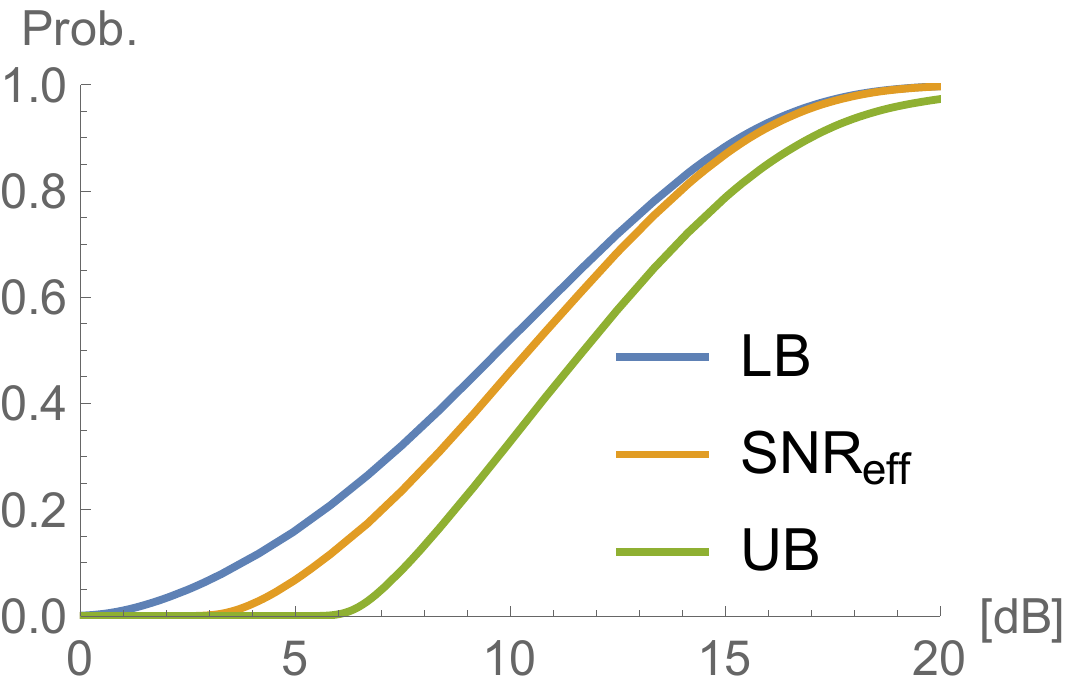}
        \caption{Lower and upper bounds for a finite upper bound.}
        \label{fig:UB_SNR_eff_Mt_2_Mr_2}
    \end{subfigure}
   	\caption{CDF of the $SNR_{eff}(\bH,\be_1)$ together with its lower and upper bounds for $M_R=M_T=2$ and when $P=10dB$.  }
	\label{fig:LBAndUB_SNR_eff_Mt_2_Mr_2}
\end{figure}

\Cref{fig:LBAndUB_SNR_eff_Mt_2_Mr_2} depicts the CDF of the effective SNR for $M_R=M_T=2$ and for $P=10 dB$ along side with its upper and lower bounds. Note that even though the same parameters where used in both sub-figures, their behavior is different, since in \cref{fig:UB_SNR_eff_Mt_2_Mr_2} we plot only the channels for which $C(\bH,P)< \infty$, i.e., $tr(\bS^{-1}) = \sum\frac{1}{d_i}< P$. Higher singular values are associated with higher capacity,  thus demanding that the sum of their \emph{inverses} is smaller than $P$ is basically enforcing usage of better channels.

%% file: tex_files/practical_IF_based_schemes/practical_IF_based_schemes_intro.tex
In this section, we present two block-based IF schemes. I.e., $\bA$ is enforced to be a block matrix. In the first scheme, the members of each block are forced to transmit using the same rate. In the second, this constraint is removed. Indeed, in order to decode the streams, successive cancellation and a nested lattice codebook are required. However, when $\bA$ is forced to be a block matrix, streams may mix only at the same block. \emph{This constraint, in addition to it's advantages in reducing the complexity of successive cancellation, reduces the search to blocks only, thus cuts down the search complexity, and, moreover, simplifies the inversion of $\bA$}.

%% file: tex_files/practical_IF_based_schemes/B_IF.tex
\subsection{Block IF (B-IF)}\label{SC:B-IF}
The vectors of $\bA$ determine the linear combinations to be decoded. Hence, the maximal block size upper bounds the number of streams mixed with each other. 
$\bA$ is forced to be a block matrix with block size of $n$. Such a scheme for $n=2$ was discussed in \cite[Section 3]{zhan2010integer}. Using the same techniques as in \cite{zhan2014integer,nazer2011compute}, we let different blocks transmit in different rates. 

The effective SNR of the $i^{th}$ block, $SNR_{eff}^i(\bH,\bA)$, is defined as the effective SNR achieved by the weakest vector of coefficient in this block. Finding the optimal $\bA$ can be done in complexity which is polynomial in both $M_T$ and $P$, instead of exponential is $M_T$ and polynomial in $P$ \cite{sahraei2015polynomially}. 

Denote the block matrix $\bA$ with blocks sizes $n$  by $\bA  = diag\{\bA_n^1,\ldots,\bA_n^k\}$, where $k = \left\lceil\frac{M_T}{n}\right\rceil$ and $\{\bA_n^i\}$ are block matrices with block size $n$. Then, the rate of the $i^{th}$ block is 
\begin{align*}
R_{B-IF}^i(\bH,\bA) = n \cdot \frac{1}{2}\log_2 SNR_{eff}^i(\bH,\bA)  =\frac{n}{2}\log_2 SNR_{eff}^i(\bH,\bA),\end{align*} 
which is the rate achieved by the worst, i.e., $n^{th}$, vector of coefficient selected in this block such that the rank of the block is full. Hence, for $K$ blocks,  the rate of B-IF is $R_{B-IF}(\bH,\bA) = \sum_{i=1}^K R_{B-IF}^i(\bH,\bA)$.

Note that an upper bound on B-IF can be achieved by taking the $n^{th}$ best vector of coefficients per block. If it is linearly independent in rest of the block, the rate of this block is achieved by the $n^{th}$ vector of coefficients. However, if they are linearly dependent, it is an upper bound, because under vector is needed, since $\bA$ must be a full ranked matrix.

 In order to construct a lower bound, we can choose the  vectors of coefficients out of a constant small set; $\{\be_1^n,\be_2^n,\be_1^n+\be_2^n,\be_1^n-\be_2^n\}$, where $\be_j^n\in \N^n$ is the $j^{th}$ vector of the standard basis. For example, if $n=2$, the set becomes $\begin{bmatrix}
1 & 0 & 1 & 1\\
0 & 1 & 1 & -1
\end{bmatrix}$. Note that any two vectors chosen from this set are linearly independent. Hence, the rate achieved by the second best vector is an achievable rate, thus, it is a lower bound on B-IF.

\begin{lemma}\label{lemma:existence_of_region_BIF_better_IF}
Let $\lambda_{max}(\bH)$ be the maximal eigenvalue of $\bH^T\bH$. Then, for any $P\leq \lambda_{max}^{-2}(\bH)$, the rate achieve by B-IF at least as good as the rate achieved by IF.
\end{lemma}

\begin{proof}[Proof of \Cref{lemma:existence_of_region_BIF_better_IF}]
According to \cref{eq:bound_the_norm_of_a}; $\|\ba_k\| < 1+ \sqrt{P}\lambda_{max}(\bH)$.
Hence, by the lemma's conditions, and since $\|\ba_k\| \in \N$, we have that $\forall k: \|\ba_k\|=1 $. Thus, $\bA=\bI$. In IF, all the senders must transmit using the same rate, i.e., the rate achieved by the worst unit vector. Note that the identity matrix is one of the blocks checked in B-IF. Hence, for the lemmas's condition, the lower bound on B-IF will choose the identity matrix for each block and $\bA$ as $\bI$. However, unlike IF, in B-IF different blocks may transmit in different rates. Thus, the only case IF and B-IF achieve the same rate is if the effective SNR for all the unit vectors is equal. If at least one of the blocks in B-IF achieves a different effective SNR, the lower bound on B-IF is better than IF, which happens almost surly.
\end{proof}

\begin{corollary}
For any realization of $\bH$, there exists a region of transmission powers such that B-IF is better than IF.
\end{corollary}

%% file: tex_files/practical_IF_based_schemes/BB_IF.tex
\subsection{Norm-Bounded-IF (NB-IF)}\label{sc:BB_IF}
In the second scheme, not only $\bA$ is restricted to be a block matrix, we also bound the norm of each row. As it turns out, at low SNR this is beneficial. Moreover, this allows us to use very simple coefficient vectors, ones for which we are able to bound the distribution of the resulting SNR, and show the superiority of this distribution to ZF and MMSE, while still keeping the complexity much lower than IF. 

Assume the block size, $n$, equals $2$. Clearly, the concept can be extended to other (small) values of $n$. Hence, for even $M_T$, there are $\frac{M_T}{2}$ 
blocks. In the suggested scheme, each block is chosen out of the following $\binom {4}{2} =6$ options: 
\begin{align*}
\begin{bmatrix} 1 & 0 \\ 0 & 1 \end{bmatrix} 
\begin{bmatrix} 1 & 1 \\ 0 & 1 \end{bmatrix}
\begin{bmatrix} 1 & 1 \\ 0 & -1 \end{bmatrix}
\begin{bmatrix} 0 & 1 \\ 1 & 1 \end{bmatrix}
\begin{bmatrix} 0 & 1 \\ 1 & -1 \end{bmatrix}
\begin{bmatrix} 1 & 1 \\ 1 & -1 \end{bmatrix}.
\end{align*}
 Thus, the norms are bounded by $2$, allowing simple yet effective linear combinations. Moreover, using the techniques in \cite{ordentlich2013successive,nazer2011compute} with this simple block structure one can easily use non equal rates and a nested lattice. Nevertheless, the most important benefit of using simple, bounded norm vectors, is the ability to analyse the effective SNR using the tools derived in \Cref{sec:bounds_on_SNR_eff}.

Specifically, in NB-IF, one chooses the best two vectors of coefficients for the $i^{th}$ block out of $\mathcal{A}_i^{M_T}$, which contains four equations. This amounts to $6$ choices. We first lower bound NB-IF by reducing this amount to four, then bound the resulting distribution. Let \small$\mathcal{A}_{LB_1}^{(2)} \defined \begin{bmatrix} 1 & 1 \\ 0 & 1 \end{bmatrix} ;  \mathcal{A}_{LB_2}^{(2)} \defined \begin{bmatrix} 0 & 1 \\ 1 & -1 \end{bmatrix}$\normalsize. Then, the vectors coefficients for each block such that the first is chosen from $\mathcal{A}_{LB_1}^{(2)}$ and the second from $\mathcal{A}_{LB_2}^{(2)}$ is a lower bound on NB-IF, and we have the following.

\begin{remark}
For any realization of $\bH$, there exists a region of transmission powers such that NB-IF is better than IF. This corollary is following directly from the proof of \Cref{lemma:existence_of_region_BIF_better_IF}.
\end{remark}
In \Cref{sbsec:IF_block_schemes}, \Cref{alg:opt_a_BB_IF} which can be used to find  $\bA_{BB-IF}^{M_T}$ matrix for NB-IF in $O(n)$ is presented.
\begin{theorem}\label{thm:CDF_of_LB_is_bounded_by_time_sharing}
Let $M_R \geq M_T \geq 2 \in \N$, $K \defined M_R-M_T+2$, $\phi \sim F_{(K,K)}$, where $F$ is the $\mathcal{F}$-distribution and $\epsilon = \epsilon(t) \in  \left(0,\frac{1}{\sqrt{2}} \right)$. Denote 
\small
\[
\rho(\epsilon) \defined 
\frac{4}{\pi} \int_0^{\frac{\pi}{4}} \left[F_\phi \left( 2(1+ \cos2 t) \cdot \left(\frac{\epsilon+1}{2\epsilon+1}\right)^2  \right) - F_{\phi}(1) \right]dt
\]
\normalsize
and $X=LB(\bH,\be_1)$, $Y=LB(\bH,\be_1-\be_2)$, where $\be_1$ and $\be_2$ are the two unit vectors. We define the condition $a_\epsilon$ as $a_{\epsilon}: Y > X(1+ \epsilon)$, and we denote the CDF of the effective SNR of MMSE and NB-IF by $F_{MMSE}(t)$ and $F_{NB-IF}(t)$ respectively. Then, 
\begin{align*}
F_{NB-IF}(t) \leq 
\min 
\left\{F_{MMSE}(t) ,  F_X(t)  - \rho(\epsilon)\left[ F_{X\vert a_\epsilon}(t)-F_{X\vert a_\epsilon}\left(\frac{t}{1+\epsilon}\right) \right] \right\}.
\end{align*}
\normalsize
\end{theorem}
The proof for \Cref{thm:CDF_of_LB_is_bounded_by_time_sharing} can be found in \Cref{subsection:proofThm3}

%% file: tex_files/practical_IF_based_schemes/results.tex
\subsection{Results}
We briefly give some numerical results and simulations to shed light on the performance of the block-based schemes. Throughout, only IF schemes are restricted to have equal rates.

\begin{figure}
	\centering
		\includegraphics[width=0.7\textwidth]{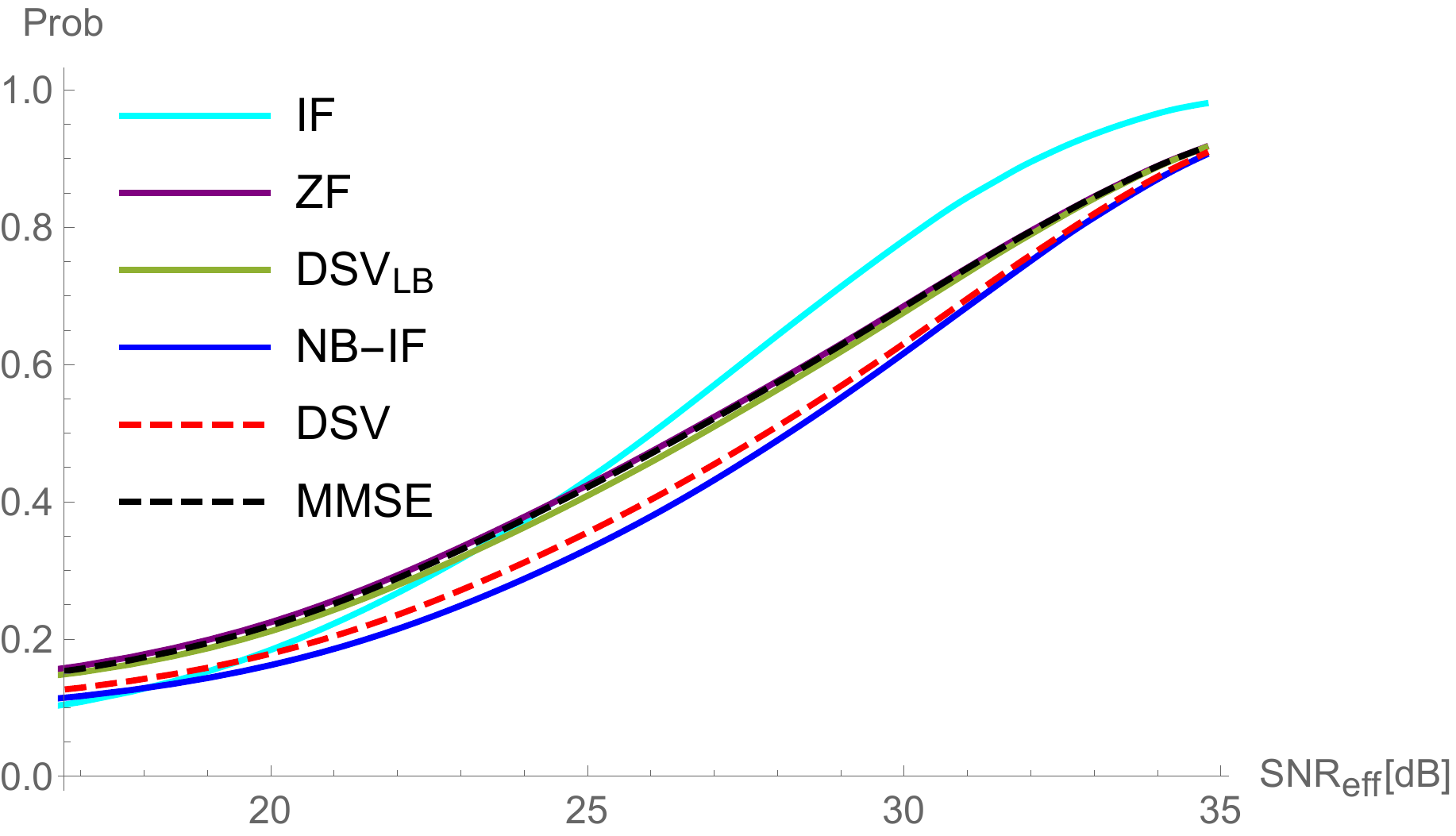}
	\caption{CDFs of IF, ZF,NB-IF, DSV, MMSE and the lower bound on DSV from \Cref{thm:CDF_of_LB_is_bounded_by_time_sharing}, where $M_R=M_T=2$. for $P=30dB$.}
	\label{fig:all_schemes_P_30db_iters_50k}
\end{figure}

\begin{figure}
	\centering
		\includegraphics[width=0.7\textwidth]{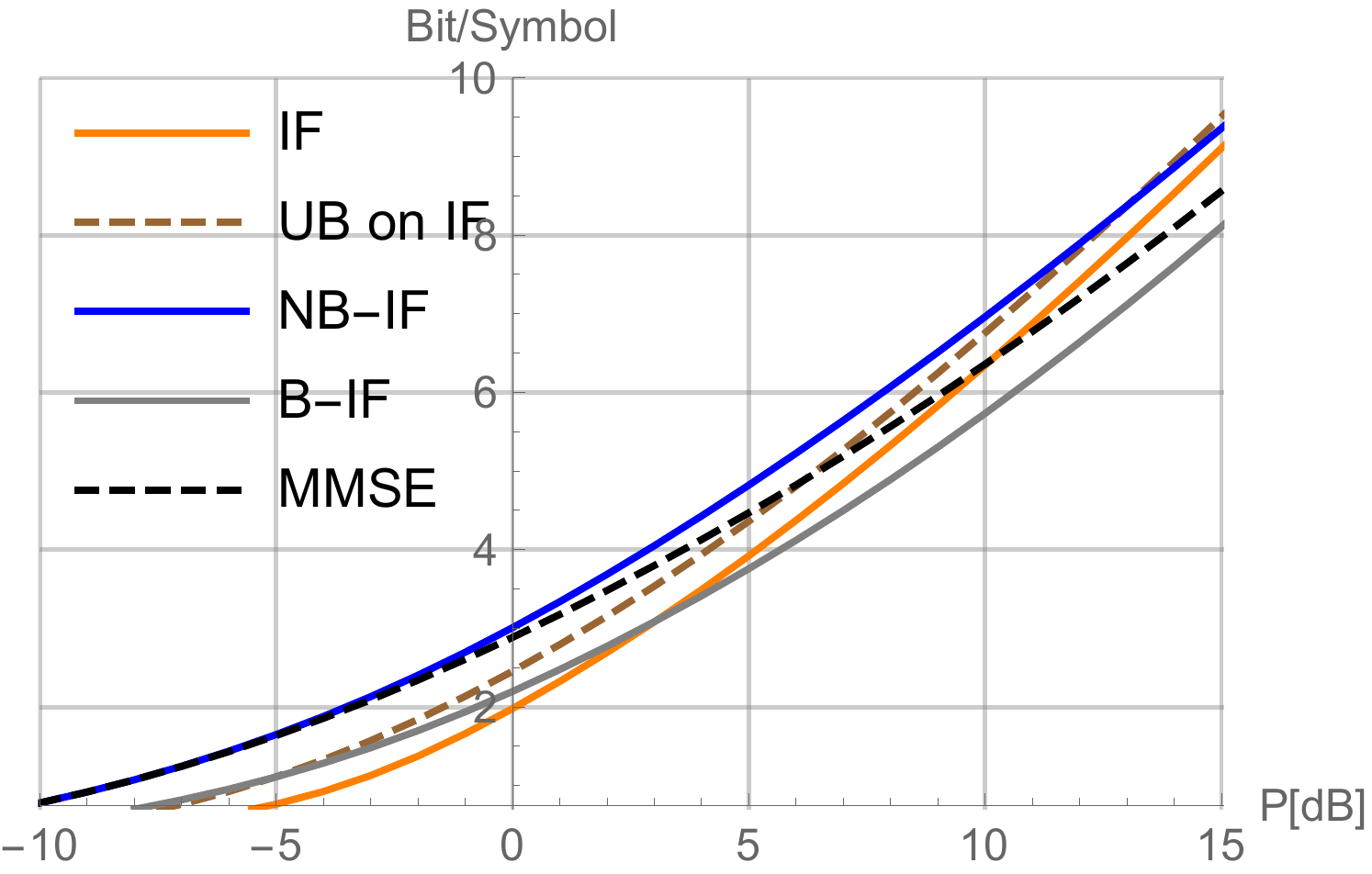}
	\caption{Average rates, i.e., the sum rate divided $M_T$, achieved by IF, B-IF, NB-IF and their bounds for $M_R=M_T=4$. The upper bound on IF (brown dashed line) is according to the upper bound in \cite{ordentlich2015precoded}.}
	\label{fig:B-IF_and_NB-IF}
\end{figure}
\Cref{fig:B-IF_and_NB-IF} depicts the performances of IF, B-IF and NB-IF together with the upper bound on IF. For any transmission power smaller than 3dB, not only is B-IF is better then IF in terms of computational complexity, performance is superior as well.
 
For lower transmission powers, we can see that NB-IF preforms better than IF. This happens because for low values of $P$, vectors with small norm are chosen by IF. However, unlike NB-IF, IF must use equal rates. For example, if both schemes use the same vectors, NB-IF has better performance. Note that when the transmission power is very small, e.g., $P<-5dB$, the graphs for MMSE and NB-IF become one. This phenomenon is not surprising, since for very low values of $P$ NB-IF selects only unit vectors, which is the same as MMSE.

At higher transmission powers, a different behavior is observed; MMSE becomes weaker and weaker relatively to the other schemes since the channel may ``support" vectors with higher norms. As long as $P<12dB$, NB-IF is better than the upper bound on IF. Thus, bounding the norm by $2$ but allowing different rates is better. For the higher values of $P$, we can see that IF becomes better and better as well. For values of $P$ higher than the ones presented, i.e., around $20dB$, IF becomes better than BB-IF. Hence, the gain from allowing the usage of higher norms is more important than the gain from transmitting in various rates.

%% file: tex_files/appendix/appendix_LB_construction.tex
\subsection{Bounds' Construction}\label{subsec:construction_of_LB} 

\begin{proof}[Proof of \Cref{lem:LB_lower_bounds_SNR_eff}]
Let \begin{align}  \bS = \bH^T\bH,\end{align}  and let $\epsilon=\frac{1}{P}$.
Note that $\bS$ is a symmetric positive-definite matrix almost surely. Hence,
\begin{align}
&\frac{SNR_{eff}(\bH,\ba)}{\left({  \ba^T\left({P\bH^T\bH}\right)^{-1}\ba }\right)^{-1}}\\
& \stackrel{(a)}{=} \frac{ \left({ \ba^T {\left( \bI+P\bH^T\bH   \right)}^{-1} \ba}\right)^{-1}}{\left({  \ba^T\left({P\bH^T\bH}\right)^{-1}\ba }\right)^{-1}}\\
&= \frac{ \left({ \ba^T {\left( \bI+ \frac{1}{\epsilon}\bS   \right)}^{-1} \ba}\right)^{-1}}{\left({  \ba^T\left({\frac{1}{\epsilon}\bS}\right)^{-1}\ba }\right)^{-1}} \\
&= \frac{\frac{1}{\epsilon} \cdot   \left({ \ba^T {\left( \epsilon \bI+ \bS   \right)}^{-1} \ba}\right)^{-1}}{\frac{1}{\epsilon} \cdot  
\left({  \ba^T{\bS}^{-1}\ba }\right)^{-1}} \\
&= \frac{\ba^T{\bS}^{-1}\ba }{\ba^T {\left( \epsilon \bI+ \bS   \right)}^{-1} \ba}. \label{eq:ratio_SNR_eff_div_LB}\\
& \stackrel{(b)}{=}
 1  + \sum_{k=1}^\infty{ \left( \epsilon \cdot \frac{ \ba^T\bS^{-1}\left(\bS + \epsilon \bI\right)^{-1}\ba}{\ba^T\bS^{-1}\ba} \right)  }^k \label{eq:ratio_eff_snr_and_lb}\\
&  \geq 1,
\end{align}
where (a) follows from \cref{eq:IF_eq_snr_eff} and (b) is due to \Cref{lemma:ratio_IF_vs_my_known_dist}.  Note that the argument in the sum is always non negative since $\epsilon$ is positive and the fraction is ratio between to quadratic forms of almost surly positive definite matrices.
\end{proof}

\begin{lemma}
\label{lemma:ratio_IF_vs_my_known_dist}
Let  $\bM$ be a symmetric positive-definite matrix and $\ba \in \R^n$ s.t $\normSqr[\ba]>0$. Then, for any $\epsilon\in \R^+$,
 \begin{align}\frac{\ba^T\bM^{-1}\ba}{\ba^T\left({ \epsilon\bI+\bM}\right)^{-1}\ba }
=
1+ \sum_{k=1}^\infty{ \left( \epsilon \cdot\frac{ \ba^T\bM^{-1}\left(\bM + \epsilon \bI\right)^{-1}\ba}{\ba^T\bM^{-1}\ba} \right)  }^k.
\end{align}
\end{lemma}
In order to prove \Cref{lemma:ratio_IF_vs_my_known_dist}, \Cref{clm:difference_btwn_invSNR_and_ainvSa,clm:UB_LB_ratio_aS2a_to_aS1a} are needed.

\begin{claim}
\label{clm:difference_btwn_invSNR_and_ainvSa}
Let  $\bM\in \R^{n \times n}$ be a symmetric positive-definite matrix and let $\ba \in \R^{n}$ s.t $\normSqr[\ba]>0$. Then, $\forall \epsilon > 0$:
\begin{align*}
 \ba^T\left({ \epsilon\bI+\bM}\right)^{-1}\ba 
=  \ba^T\bM^{-1}\ba  -  \epsilon \cdot \ba^T\bM^{-1}\left(\bM + \epsilon \bI\right)^{-1}\ba
\end{align*}
\end{claim}
\begin{proof}[Proof of \Cref{clm:difference_btwn_invSNR_and_ainvSa}]
By SVD Let $\bM = \bU^T\bD\bU$, let $d_i \defined [\bD]_{ii}$ and let $\bw \defined \bU\ba$. Then,
\begin{align*}
\ba^T\bM^{-1}\ba  - \ba^T\left({ \epsilon\bI+\bM}\right)^{-1}\ba 
&= \ba^T \left(\bM^{-1} -{\left({ \epsilon\bI+\bM}\right)^{-1}} \right) \ba \\
&= \ba^T \left(  (\bU^T\bD\bU)^{-1} -{\left({ \epsilon\bI+\bU^T\bD\bU}\right)^{-1} } \right) \ba \\
&= \ba^T\bU^T \left( \bD^{-1} - (\epsilon \bI + \bD)^{-1}  \right) \bU\ba\\
&=\bw^T \left( \bD^{-1} - (\epsilon \bI + \bD)^{-1}  \right) \bw \\
&= \sum_{i=1}^n {\left(\frac{w_i^2}{[\bD]_{ii}} -  \frac{w_i^2}{\epsilon + [\bD]_{ii}}  \right)} \\
&= \sum_{i=1}^n \frac{\epsilon w_i^2}{[\bD]_{ii} \left([\bD]_{ii} +\epsilon \right)} \\
&= \epsilon \sum_{i=1}^n \frac{ w_i^2}{d_i \left(d_i +\epsilon \right)}\\
&= \epsilon \cdot \bw^T \left(\bD^{-1} (\bD+\epsilon\bI)^{-1} \right)\bw\\
&= \epsilon \cdot \ba^T\bU^T \bD^{-1} (\bD+\epsilon\bI)^{-1} \bU\ba\\
&= \epsilon \cdot \ba^T\bM^{-1}\left(\bM + \epsilon \bI\right)^{-1}\ba.
\end{align*}
\end{proof}

\begin{corollary}
 $ \ba^T\bM^{-1}\ba >  \ba^T\left({ \epsilon\bI+\bM}\right)^{-1}\ba > 0$ \label{crol:minum1_of_aSa_is_LB}
\end{corollary}
\begin{corollary}
 $ \ba^T\bM^{-1}\ba >  \epsilon \cdot \ba^T\bM^{-1}\left(\bM + \epsilon \bI\right)^{-1}\ba > 0 $ \label{crol:aSa_bigger_then_delta_of_epsilon}
\end{corollary}

\begin{definition}
Let $\bM\in\R^{n \times n}$ be  a symmetric positive-definite symmetric matrix, and let $d_i$ be the $i^{th}$ singular value of $\bM$. \end{definition}

Using SVD, $\bM= \bU^T\bD\bU$ where $\bU$ is a matrix built from $\bM$'s eigenvectors and $\bD = diag(\{d_i\})$. 
\begin{remark}\label{rmrk:sum_of_1_dev_SV_is_trance_of_inv_mat}
\begin{equation*}
\sum_{i=1}^n {1/d_i} = tr(\bM^{-1}).
\end{equation*}
\end{remark}

\begin{claim}\label{clm:UB_LB_ratio_aS2a_to_aS1a}
Let $\bM\in \R^{n \times n}$ be a symmetric positive-definite matrix. Then, for any $\ba \in \R^n$ s.t $\normSqr[\ba]>0$,
\begin{equation*}
 \frac{\ba^T\bM^{-2}\ba}{\ba^T\bM^{-1}\ba} \leq  tr(\bM^{-1}).
\end{equation*}
\end{claim}
\begin{proof}[Proof of \Cref{clm:UB_LB_ratio_aS2a_to_aS1a}]
\begin{align*}
\frac{\bu_k^T\bM^{-2}\bu_k}{\bu_k^T\bM^{-1}\bu_k} &= \frac{\frac{1}{d_k^2}}{\frac{1}{d_k}} = \frac{1}{d_k}
\end{align*}
Since $span(\bU) = \bR^n$, $\ba$ can be written as a linear combination of the vectors in $\bU$. Let $\sum_{i=1}^n {\alpha_i\bu_i} = \ba$, then
\begin{align*}
\frac{\ba^T\bM^{-2}\ba}{\ba^T\bM^{-1}\ba} 
&= \frac{(\sum_{i=1}^n {\alpha_i\bu_i})^T\bM^{-2}(\sum_{i=1}^n {\alpha_i\bu_i})}{(\sum_{i=1}^n {\alpha_i\bu_i})^T\bM^{-1}(\sum_{i=1}^n {\alpha_i\bu_i})} \\
&=\frac{\sum_{i=1}^n{ \frac{\alpha_i^2 }{d_i^2}}}{\sum_{i=1}^n{ \frac{\alpha_i^2 }{d_i}}} \cdot \frac{1/\normSqr[\ba]}{1/\normSqr[\ba]}\\ 
&=
 \frac{\sum_{i=1}^n{ \frac{\alpha_i^2 }{\normSqr[\ba] } \cdot \frac{1}{d_i^2}}}{\sum_{i=1}^n{ \frac{\alpha_i^2 }{\normSqr[\ba] } \cdot \frac{1}{d_i}}} \\
& = \frac{\sum_{i=1}^n{ \left(\frac{\alpha_i^2}{\sum_j \alpha_j^2}\right) \cdot \frac{1}{d_i^2}}}{\sum_{i=1}^n{ \left(\frac{\alpha_i^2}{\sum_j \alpha_j^2}\right) \cdot \frac{1}{d_i}}}\\
& = \frac{\sum_{i=1}^nf_i/d_i^2}{\sum_{i=1}^nf_i/d_i}\\
& \stackrel{(a)}{\leq} tr(\bM^{-1}),
\end{align*}
where $(a)$ follows from  \cref{rmrk:sum_of_1_dev_SV_is_trance_of_inv_mat} and  using Cauchy--Schwarz inequality with the sequences  $a_i = \frac{f_i}{d_i}$ and $b_i=\frac{1}{d_1}$.
\end{proof}

\begin{corollary}\label{corl:bound_by_trace}
 \begin{align*}
\frac{\ba^T\bM^{-1}\left(\bM + \epsilon \bI\right)^{-1}\ba}{\ba^T\bM^{-1}\ba} 
 &\stackrel{(a)}{\leq} \frac{\ba^T\bM^{-1}\bM^{-1}\ba}{\ba^T\bM^{-1}\ba} \\
&= \frac{ \ba^T\bM^{-2}\ba}{\ba^T\bM^{-1}\ba}\\ 
 &\stackrel{(b)}{\leq} tr(\bM^{-1})
,
\end{align*}
where (a) is following from \Cref{clm:difference_btwn_invSNR_and_ainvSa} and (b) follows from \Cref{clm:UB_LB_ratio_aS2a_to_aS1a}.
\end{corollary}

Now we can prove \Cref{lemma:ratio_IF_vs_my_known_dist}.
\begin{proof}[Proof of \Cref{lemma:ratio_IF_vs_my_known_dist}]

 Using \Cref{corl:bound_by_trace}, let as define 
\begin{equation*}
\delta(\epsilon) \defined \epsilon \cdot \frac{\ba^T\bM^{-1}\left(\bM + \epsilon \bI\right)^{-1}\ba}{\ba^T\bM^{-1}\ba}.
\end{equation*} 
Thus,
\begin{align*}
\frac{\ba^T\bM^{-1}\ba}{\ba^T\left({ \epsilon\bI+\bM}\right)^{-1}\ba } 
&= \frac{\ba^T\bM^{-1}\ba} {\ba^T\bM^{-1}\ba - \epsilon \cdot \ba^T\bM^{-1}\left(\bM + \epsilon \bI\right)^{-1}\ba}\\ 
&= 1+ \frac{\epsilon \cdot \ba^T\bM^{-1}\left(\bM + \epsilon \bI\right)^{-1}\ba}{\ba^T\bM^{-1}\ba - \epsilon \cdot \ba^T\bM^{-1}\left(\bM + \epsilon \bI\right)^{-1}\ba}\\
&= 1+\frac{ \epsilon \cdot\frac{ \ba^T\bM^{-1}\left(\bM + \epsilon \bI\right)^{-1}\ba}{\ba^T\bM^{-1}\ba}}{1-  \epsilon \cdot\frac{ \ba^T\bM^{-1}\left(\bM + \epsilon \bI\right)^{-1}\ba}{\ba^T\bM^{-1}\ba}}\\
&= 1+\frac{\delta(\epsilon)}{1-\delta(\epsilon)}.
\end{align*}
According to \Cref{crol:aSa_bigger_then_delta_of_epsilon}, $0 < \delta(\epsilon) < 1$, therefore, by the sum of geometric series  we have:
\begin{align*}
\frac{\ba^T\bM^{-1}\ba}{\ba^T\left({ \epsilon\bI+\bM}\right)^{-1}\ba } 
&=  1+ \sum_{k=1}^\infty{\delta(\epsilon)}^k\\
&= 1+ \sum_{k=1}^\infty{ \left(\epsilon \cdot \frac{ \ba^T\bM^{-1}\left(\bM + \epsilon \bI\right)^{-1}\ba}{\ba^T\bM^{-1}\ba} \right)  }^k.
\end{align*}
\end{proof}

\begin{claim}\label{clm:LB_ratio_IF_vs_my_known_dist}
Let $M\in\R^{n \times n}$ be a symmetric positive definite matrix. Then, for any $\ba \in \R^n$ s.t $\normSqr[\ba]>0$,
\begin{equation*}
\frac{\ba^T\bM^{-1}\ba}{\ba^T\left({ \epsilon\bI+\bM}\right)^{-1}\ba }  
\leq 
1+ \sum_{k=1}^\infty{\left(  \epsilon \cdot tr(\bM^{-1}) \right)}^k.
\end{equation*}
\end{claim}
\begin{proof}[Proof of \Cref{clm:LB_ratio_IF_vs_my_known_dist}]
\begin{align}
\frac{\ba^T\bM^{-1}\ba}{\ba^T\left({ \epsilon\bI+\bM}\right)^{-1}\ba } 
&\stackrel{(a)}{=} 1+ \sum_{k=1}^\infty{ \left( \epsilon \cdot\frac{ \ba^T\bM^{-1}\left(\bM + \epsilon \bI\right)^{-1}\ba}{\ba^T\bM^{-1}\ba} \right)  }^k \\
&\stackrel{(b)}{\leq} 1+ \sum_{k=1}^\infty{ \left(\epsilon \cdot  \frac{\ba^T\bM^{-2}\ba}{\ba^T\bM^{-1}\ba} \right)  }^k \\
& \stackrel{(c)}{\leq} 
1+ \sum_{k=1}^\infty{\left( \epsilon \cdot tr(\bM^{-1})  \right)}^k,
\end{align} 
where $(a)$ follows from \Cref{lemma:ratio_IF_vs_my_known_dist}, (b) follows from \Cref{crol:minum1_of_aSa_is_LB} and (c) is by \Cref{clm:UB_LB_ratio_aS2a_to_aS1a}.
\end{proof}

\begin{claim}\label{clm:find_P_st_C_smaller_then_t_WHP}
Let $d_{min}=d_{M_T}$ denote the minimal singular value of $\bS$, and let $F_{d_{min}}(t)$ denote the CDF of $d_{min}$, then, for any $P$ and $t$
\begin{align}
\Pr\left\{C(\bH,P)\leq t\right\}   \geq  u(t-1)\left[1- F_{d_{min}}\left(  \frac{M_T}{P\left(1-\frac{1}{t}\right)}  \right)\right].
\end{align}
 Hence, $\forall t>1$, a transmission power $P$ such that 
\begin{equation*}\Pr\{C(\bH,P)>t\}\end{equation*} 
is negligible can be found.
\end{claim}

\begin{figure}
		\includegraphics[width=0.60\textwidth]{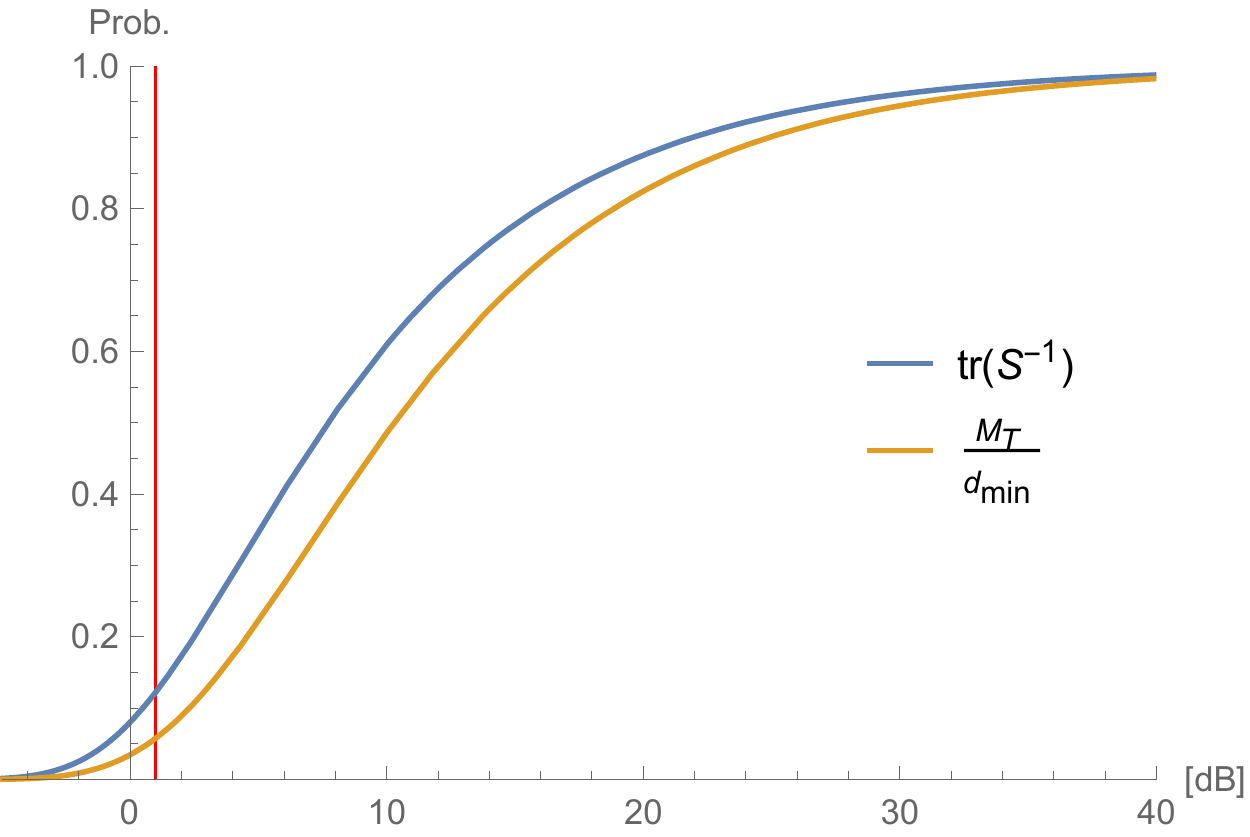}
	\caption{The distributions from simulations of $tr(\bS^{-1})$ and $\frac{M_T}{d_{M_T}}$ for $M_R=M_T=2$ and $P=1$ . The CDF behavior gets along with \cref{eq:bound_Tr_invS}. The vertical red line is the $P=1$.}
	\label{fig:nr_2_nt_2_combined_cdf}
\end{figure}

 \begin{proof}[Proof of \Cref{clm:find_P_st_C_smaller_then_t_WHP}]
Note that 
\begin{align}\label{eq:bound_Tr_invS}
tr(\bS^{-1}) = \sum_{i=1}^{M_T}\frac{1}{d_i} \leq \frac{M_T}{d_{M_T}}.
\end{align} 
Thus,
\begin{align*}
\Pr\left\{C(\bH,P)\leq t\right\} 
&\stackrel{(a)}{=} \Pr\left\{\sum_{k=0}^\infty \left(\frac{tr(\bS^{-1})}{P}\right)^k\leq t\right\} \\
&  \stackrel{(b)}{\geq}\Pr\left\{\sum_{k=0}^\infty \left(\frac{M_T}{P d_{min}}\right)^k\leq t\right\} \\
&= \p[\frac{M_T}{P d_{min}}<1]\p[\sum_{k=0}^\infty \left(\frac{M_T}{P d_{min}}\right)^k\leq t  \big{\vert}\frac{M_T}{P d_{min}}<1] \\
& +  \Pr\left\{ \frac{M_T}{P d_{min}}\geq 1 \right\}\underbrace{\Pr\left\{\sum_{k=0}^\infty \left(\frac{M_T}{P d_{min}}\right)^k\leq t  \big{\vert}\frac{M_T}{P d_{min}}\geq 1\right\}}_{0} \\
&\stackrel{(c)}{=}   \Pr\left\{ \frac{M_T}{P d_{min}}<1 \right\}\Pr\left\{\frac{1}{1-\frac{M_T}{P d_{min}}}\leq t  \big{\vert}\frac{M_T}{P d_{min}}<1\right\}\\
& = \Pr\left\{ d_{min}>\frac{M_T}{P} \right\} \Pr\left\{d_{min} \geq \frac{M_T}{P\left(1-\frac{1}{t}\right)}\big{\vert} d_{min}>\frac{M_T}{P} \right\}\\
&= \Pr\left\{d_{min} \geq \frac{M_T}{P\left(1-\frac{1}{t}\right)},d_{min}>\frac{M_T}{P} \right\}\\
&= \Pr\left\{d_{min} \geq \max\left(\frac{M_T}{P\left(1-\frac{1}{t}\right)},\frac{M_T}{P}\right) \right\}\\
&\stackrel{(d)}{\geq} u(t-1) \Pr\left\{d_{min} \geq \max\left(\frac{M_T}{P\left(1-\frac{1}{t}\right)},\frac{M_T}{P}\right) \right\}\\
&\stackrel{(e)}{=} u(t-1)  \Pr\left\{d_{min} \geq \frac{M_T}{P\left(1-\frac{1}{t}\right)} \right\}\\
&= u(t-1)\left[1- F_{d_{min}}\left(  \frac{M_T}{P\left(1-\frac{1}{t}\right)}  \right)\right],
\end{align*}
where (a) follows from \cref{eq:def_of_C}; (b) is due to \cref{eq:bound_Tr_invS}; (c) is by the sum of geometric series ; (d) is since the probability is a non negative function and (e) is because when $t\geq 1$, the left expression in the maximum function is no less then the right one.
\end{proof}	
\Cref{clm:find_P_st_C_smaller_then_t_WHP} can be exploited by assuming that at most two out of: $\{P,M_T,M_R\}$ are constant, and the non constant ones can be chosen. Since the distribution of the minimal singular value in \cite[Lemma 4.1]{edelman1988eigenvalues} is dependent in both $M_T$ and $M_R$ while \Cref{clm:find_P_st_C_smaller_then_t_WHP} is dependent also in $P$, the non constant parameters can always be chosen to fit any desired negligible probability.

\begin{claim}\label{claim:dist_of_matrix_form_LB}
Let $\mathcal{A}$ be a set of vectors. Then, Given $\mathcal{A}$, the lower bound has a Wishart distribution with the following parameters,
\begin{equation}
LB(\bH,\mathcal{A}) \sim \mathcal{W}(P({\mathcal{A}^T\mathcal{A}})^{-1},rank(\mathcal{A}),M_R-M_T+rank(\mathcal{A})).
\label{eq:dist_LB_eff_SNR_matrix_form}
\end{equation}
\end{claim}

\begin{proof}[Proof of \Cref{claim:dist_of_matrix_form_LB}]
\begin{align}
&LB(\bH,\mathcal{A}) 
= P \cdot{\left(\mathcal{A}^T{(\bH\bH^T)}^{-1}\mathcal{A}\right)}^{-1}\\
&= P \cdot({\mathcal{A}^T \bS^{-1} \mathcal{A}})^{-1}\\ 
&\stackrel{(c)}{\sim} P \cdot\mathcal{W}(({\mathcal{A}^T\bSig^{-1}\mathcal{A}})^{-1},rank(\mathcal{A}),M_R-M_T+rank(\mathcal{A}))\\
&\sim P \cdot\mathcal{W}(({\mathcal{A}^T\bI\mathcal{A}})^{-1},rank(\mathcal{A}),M_R-M_T+rank(\mathcal{A}))\\
&\sim\mathcal{W}(P({\mathcal{A}^T\mathcal{A}})^{-1},rank(\mathcal{A}),M_R-M_T+rank(\mathcal{A})),
\end{align} 
where (c) is according to \cite[proposition 8.9]{eaton}.
\end{proof}

%% file: tex_files/appendix/Practical_IF_Based_Schemes.tex
\subsection{Practical IF-Based Schemes}\label{sbsec:IF_block_schemes}

Let us denoted the zeros matrix with $a$ rows and $b$ columns by $O^{a \times b}$. And let the matrix $\mathcal{A}_i^{M_T}$ be 
\begin{align}
\mathcal{A}_i^{M_T} \defined 
\begin{bmatrix}
O^{k_1 \times 4};
\mathcal{A}^{(2)};
O^{k_2 \times 4} 
\end{bmatrix} \in \Z^{M_T \times 4}, 
\end{align}
where $k_1=2(i-1)$ and $k_2=M_T-2i$.

For example, by setting $M_T = 4$, we have that
\begin{align*}
\mathcal{A}_{1}^4 = 
\begin{bmatrix}
1 & 0 & 1 & 1\\
0 & 1 & 1 & -1\\
0 & 0 & 0 & 0\\
0 & 0 & 0 & 0
\end{bmatrix} 
\text{ , } 
\mathcal{A}_{2}^4 = 
\begin{bmatrix}
0 & 0 & 0 & 0\\
0 & 0 & 0 & 0\\
1 & 0 & 1 & 1\\
0 & 1 & 1 & -1
\end{bmatrix}.
\end{align*}
\begin{algorithm}
\caption{Find $\bA_{BB-IF}$ for Even $M_T$}
\label{alg:opt_a_BB_IF}
\SetKwFunction{Return}{return}\SetKwInOut{Input}{input}\SetKwInOut{Output}{output}
\LinesNumbered 

\Input{$\bH$ s.t $M_T \bmod 2 =0$}
\Output{$\bA_{BB-IF}$, a matrix for NB-IF}

$\bA_{NB-IF} \leftarrow  \emptyset$ \;
$i \leftarrow 1$\;
\While {$2i\leq M_T$}{
$\ba_i \leftarrow $  $\underset{\ba\in \mathcal{A}_{i}^{M_T}}{\operatorname{argmax}} \text{ } SNR_{eff}(\bH,\ba)$ // find best\;
$\ba_{i+1} \leftarrow $  $\underset{\ba\in \mathcal{A}_{i}^{M_T} \setminus \ba_i}{\operatorname{argmax}} \text{ } SNR_{eff}(\bH,\ba)$ // find $2^{nd}$ best\;
$\bA_{NB-IF} \leftarrow \left[\bA_{NB-IF}, \ba_i,\ba_{i+1}\right]$\; 
$i \leftarrow i+2$\;
}
\Return $\bA_{NB-IF}$\;

\end{algorithm}\DecMargin{1em}
\begin{remark}\label{rem:extension_to_3_size_block} 
\Cref{alg:opt_a_BB_IF} can be extended to the odd case by noting that every odd number $k \geq 3$ can be written as $k=2i+1 = 3 + 2(i-1)$, where $i \in \N$. Hence, any block matrix is built from $i-1$ blocks of $2 \times 2$ and one block of $3 \times 3$. The vectors for the last block are the three best linearly independent vector from \cref{eq:lastBlock3x3vectors}, with one or two non zero elements. Each non zero element is either $1$ or $-1$. I.e.,
\begin{equation}\label{eq:lastBlock3x3vectors}
\left\{\ba : \normSqr[\ba]\in \{1,2\}, \sum_{i=M_T-2}^{M_T}\normSqr[\ba^T \be_{i}]>0   \right\}.
\end{equation} 
\end{remark}

\begin{remark}
By choosing $\bA=\bI$ and allowing different rates, the MMSE is achieved, \cite[Section 3B]{zhan2014integer}. Note that the lower bound on NB-IF is lower bounded by the MMSE, since the identity matrix is one of the options being checked.
\end{remark}

%% file: tex_files/appendix/proof_of_the_lb_theorem.tex
\subsection{Proof Theorem \ref{thm:CDF_of_LB_is_bounded_by_time_sharing}}\label{subsection:proofThm3}
In order to prove \Cref{thm:CDF_of_LB_is_bounded_by_time_sharing} the performance of NB-IF is lower bounded. It is done by upper bounding the CDF of NB-IF using a new scheme, Distributed Selection of Vectors (DSV). Then, we show that DSV is lower bounded by the MMSE and an upgraded version of the ZF which lower bounds DSV as well.

\subsubsection{Distributed Selection of Vectors (DSV)}\label{subsc:DSV}
This scheme is based on distributed selection of vectors (DSV). I.e, the matrix $\bA$ is built by choosing the best $\ba_i$'s vector of coefficients out of the $i^{th}$ set; a set containing two vectors, where $i \in \{1,\ldots,M_T\}$. The vectors' sets are constructed such that no matter which two are chosen for $\bA$. $\bA$'s rank is always full, i.e., 2.

In NB-IF the decoder is allowed to use any block out of
$\begin{bmatrix} 1 & 0 \\ 0 & 1 \end{bmatrix} 
\begin{bmatrix} 1 & 1 \\ 0 & 1 \end{bmatrix}
\begin{bmatrix} 1 & 1 \\ 0 & -1 \end{bmatrix}
\begin{bmatrix} 0 & 1 \\ 1 & 1 \end{bmatrix}
\begin{bmatrix} 0 & 1 \\ 1 & -1 \end{bmatrix}
\begin{bmatrix} 1 & 1 \\ 1 & -1 \end{bmatrix}$ 
Note that row permutations  cause to  a permutation of the same rates vector and multiply a vector of coefficients by $-1$ does not change the rate at all, \cref{eq:IF_eq_snr_eff}, we can add more matrices to the NB-IF without changing its rate. When more than one matrix achieve the same rate, we randomly select one of them. Thus, we denoted  all the options for a single block in NB-IF by:
\begin{equation}
\mathcal{A}_{NB-IF} \defined \left\{\bA_i^2,\bA_i^{2*}\right\}_{i=1}^6,
\label{eq:def_of_As}
\end{equation} 
where
\begin{align*}
&\bA_1^2=\begin{bmatrix}
   1 & 0 \\
   0 & 1
\end{bmatrix}
\bA_2^2=\begin{bmatrix}
   1 & 1 \\
   0 & 1
\end{bmatrix}
\bA_3^2 =\begin{bmatrix}
   1 & 1 \\
   0 & -1
\end{bmatrix}\\
&\bA_4^2 =\begin{bmatrix}
   0 & 1 \\
   1 & 1
\end{bmatrix}
\bA_5^2 =\begin{bmatrix}
   0 & 1 \\
   1 & -1
\end{bmatrix}
\bA_6^2 =\begin{bmatrix}
   1 & 1 \\
   1 & -1
\end{bmatrix}                        \\
&\bA_1^{2*}=\begin{bmatrix}
   0 & 1 \\
   1 & 0
\end{bmatrix}
\bA_2^{2*}=\begin{bmatrix}
  1 & 1 \\
	1 & 0
	\end{bmatrix}
\bA_3^{2*}=\begin{bmatrix}
  1 & 1\\
 -1 & 0
\end{bmatrix}\\
&\bA_4^{2*} =\begin{bmatrix}
   1 & 0 \\
   1 & 1
\end{bmatrix}
\bA_5^{2*} =\begin{bmatrix}
   1 & 0 \\
   -1 & 1
\end{bmatrix}
\bA_6^{2*}=\begin{bmatrix}
  1 &1\\
 -1 & 1.
\end{bmatrix}
\end{align*}
Note that $\forall i$: $\bA_i^2$ and $\bA_i^{2*}$  are permutations of each other, because $\bA_i^{2*} = \bA_i^2 \cdot \begin{bmatrix}  0 & 1\\ 1 & 0\end{bmatrix}$, therefore, achieve a permutation of the same rates vector.

Let
\begin{align}
\ba_{i,1}^{M_T} &= \be_i^{M_T} \in \Z^{M_T} \label{eq:def_a_i_1}\\
\ba_{i,2}^{M_T} &= (-1)^{i+1}\be_i^{M_T} + \be_{i+(-1)^{i+1}}^{M_T} \in \Z^{M_T} \label{eq:def_a_i_2}\\
\mathcal{A}_i^{M_T}&= \{\ba_{i,1}^{M_T},\ba_{i,2}^{M_T} \} \in \Z^{M_T \times 2} \label{eq:def_A_i}, 
\end{align}
and let 
\begin{align}\label{eq:mathcal_A_def}
\mathcal{A}_{DSV}^{M_T} &\defined \mathcal{A}_1^{M_T} \times \mathcal{A}_2^{M_T} \cdots \times \mathcal{A}_n^{M_T}  \cdots \times \mathcal{A}_{M_T}^{M_T}\\
\bA^{M_T} &\in  \mathcal{A}_{DSV}^{M_T}.
\end{align}
Note that 
\begin{align}\label{eq:mathcalAA}
&\mathcal{A}_i^{M_T}  =\begin{bmatrix}
   0  &  0\\
	 \vdots & \vdots\\
   \frac{1}{2}\left(1+(-1)^{i+1}\right) & 1 \\
   \frac{1}{2}\left(1+(-1)^{i}\right) & (-1)^{i+1}\\
	 \vdots & \vdots\\
	0  &  0
\end{bmatrix}\\ 
& \Longrightarrow
 \left(\mathcal{A}_i^{M_T}\right)^T\mathcal{A}_i^{M_T}  =\begin{bmatrix}
  1 					& (-1)^{i+1} \\
  (-1)^{i+1} 	& 2
\end{bmatrix},
\end{align}
hence,
any $\bA^{M_T}$ chosen from $\mathcal{A}_{DSV}^{M_T}$, is a full rank matrix. Equivalently, we can see that the first and second column in each block can be chosen independently from $\mathcal{A}_1$ and $\mathcal{A}_2$ respectively, where
\begin{align*}
\mathcal{A}_1 &= \left\{\begin{bmatrix}1 \\ 0\end{bmatrix}, \begin{bmatrix}1 \\ 1\end{bmatrix}\right\} \text{ , }
\mathcal{A}_2 = \left\{\begin{bmatrix}0 \\ 1\end{bmatrix}, \begin{bmatrix}1 \\ -1\end{bmatrix}\right\}.
\end{align*}
Hence, the entire block can be chosen from
\begin{align*}
\mathcal{A}_1^2 \times \mathcal{A}_2^2 &=
\left\{
\begin{bmatrix}
   1 & 0 \\
   0 & 1
\end{bmatrix},
\begin{bmatrix}
   1 & 1 \\
   0 & -1
\end{bmatrix},
\begin{bmatrix}
   1 & 0 \\
   1 & 1
\end{bmatrix},
\begin{bmatrix}
   1 & 1 \\
   1 & -1
\end{bmatrix}      
\right\}\\
& = \left\{\bA_1^2,\bA_3^2,\bA_4^{*6},\bA_6^2   \right\}.
\end{align*}  
Since columns permutation only changes the order of the elements in the effective SNR vectors, we can say that this scheme is equivalent to choosing the best out of $\mathcal{A}_{DSV}$ for each block, where
\begin{align}
\mathcal{A}_{DSV} &\defined diag\left(\{\bA_1^2,\bA_1^{*2},\bA_3^2,\bA_3^{*2},\bA_4^2,\bA_4^{2*},\bA_6^2,\bA_6^{*2}   \}\right).  
\end{align}
For example, 
\begin{align*}
diag\left(\left\{\begin{bmatrix}
   1 & 2 \\
   3 & 4
\end{bmatrix},
\begin{bmatrix}
   5 & 6 \\
   7 & 8
\end{bmatrix},9\right\}\right) = 
\begin{bmatrix}
   1 & 2 &0 &0 &0 \\
   3 & 4 &0 &0 &0\\
	 0 & 0 &5 &6 &0\\
	 0 & 0 &7 &8 &0\\
	 0 & 0 &0 &0 &9
 \end{bmatrix}.
\end{align*}
Note that according to \cref{eq:def_of_As}, $\mathcal{A}_{DSV}\subset\mathcal{A}_{NB-IF}$. Hence, every achievable rate by DSV, is also achievable by NB-IF. For this reason, DSV lower bounds NB-IF.

\subsubsection{Proof Steps}\label{ssc:proof_steps}
The proof of \Cref{thm:CDF_of_LB_is_bounded_by_time_sharing} in based on DSV. To that end, we go through the following steps:
\begin{enumerate}
	\item  Show that DSV is least at good as the MMSE.
	\item  Showing that in DSV, the rate achieved by each block is distributed the same. Hence, we can focus our attention to a single block.
	\item  Lower bounding the effective SNR achieved by a selected vector of coefficients, and show that even the lower bound on DSV preforms at least the good as the maximum between the MMSE and an upgraded version of ZF.
	\item  Recalling that DSV is upper bounded by NB-IF, hence, the value of its CDF is never smaller than the CDF of NB-IF.
\end{enumerate}

\subsection*{Step 1:}
\begin{claim}\label{clm:MMSE_is_LB_to_DSV}
The MMSE is a lower bound for DSV.
\end{claim}
\begin{proof}[Proof of \Cref{clm:MMSE_is_LB_to_DSV} ]
On one hand, we know that the MMSE is achieved by identity matrix, i.e., when $\bA=\bI$, as can be found in \cite{zhan2014integer}. On the other hand, note that if all the blocks are chosen as $\bI_{2 \times 2}$, then $\bA=\bI$. Therefore, the identity matrix is always checked as an option when DSV is being used. Accordingly, any rate achieved by MMSE can be achieved by DSV.
\end{proof}

\subsection*{Step 2:}
In this step we show that the index of the block in $\bA$ does not effect the distribution of the lower bound on the effective SNR. I.e., all the blocks are distributed the same. It is done by showing that the lower bound on even blocks and on on odd blocks is distributed the same. Then, we show that any channel with $(M_T,M_R)$ (transmitting , receiving) antennas can be reduced to a channel with $(2 , M_R-M_T+2)$ antennas, as long as $M_T\geq 2$ without changing the lower bound's behavior. Finally we prove that when $M_T=2$, the lower bound on the effective SNR of the first and second  blocks are identical, concluding that all the blocks can be lower bounded using the same distribution.

\begin{lemma}[Channel reduction]\label{lem:LB_for_each_eq_set_dist_identical}
Let  $2 \leq M_T \leq  M_R  \in \N$, $j \in \{1,\ldots,M_T\}$ and let
\begin{align*}
\begin{split}
\bH &\in \R^{M_T \times M_R}  \text { ; } \bS = \bH\bH^T\\
\tilde{\bH} &\in \R^{2 \times (M_R-M_T+2)} \text { ; } \tilde{\bS} = \tilde{\bH}\tilde{\bH}^T.
\end{split}
\end{align*}
be independent channels such that each element of $\bH,\tilde{\bH} \sim \mathcal{N}(0,1)$ i.i.d. Then, for any 
\begin{align*}
\tilde{\ba}_1 &= [a_{11},a_{12}]^T \in \R^2 \\ 
\tilde{\ba}_2 &=[a_{21},a_{22}]^T \in \R^2\\
\ba_1 &= [0,\ldots 0, a_{11},a_{12}, 0,\ldots 0 ]^T \in \R^{M_T} \\ 
\ba_2 &=[0,\ldots 0, a_{21},a_{22},0,\ldots 0 ]^T \in \R^{M_T}
\end{align*}
\begin{equation*}
\forall i\in\{1,2\}: LB(\bH,\ba_i)   \sim  LB(\tilde{\bH},\tilde{\ba}_i)
\end{equation*}
where $LB(\bH,\ba)$ is by \Cref{def:LB_def}.
\end{lemma}
In order to prove \Cref{lem:LB_for_each_eq_set_dist_identical}, \Cref{claim:one_eq_behavior_depends_only_in_Mr_minus_Mt,claim:channel_reduction,clm:first_and_sec_eq_set_behave_the_same} are needed.
\begin{claim}[Reduction per vector]\label{claim:one_eq_behavior_depends_only_in_Mr_minus_Mt}
\begin{equation*}
\forall i\in\{1,2\}: LB(\bH,\ba_i) \sim LB(\tilde{\bH},\tilde{\ba}_i).
\end{equation*}
\end{claim}
\begin{proof}[Proof of \Cref{claim:one_eq_behavior_depends_only_in_Mr_minus_Mt}]
For convenience let us denote $\tilde{M}_R=M_R-M_T+2$ and $\tilde{M}_T=2$. Then, according to \Cref{corl:LB_distribution},
\small\begin{align*}
&LB(\bH,\ba_i) \sim  \Gamma\left(\frac{M_R-M_T+1}{2},\frac{2P}{\normSqr[\ba_i]}\right)\\
 &\sim  \Gamma\left(\frac{[M_R-(M_T-2)]-[M_T-(M_T-2)]+1}{2},\frac{2P}{\normSqr[\ba_i]}\right)\label{eq:only_idx_matters}\\
&\sim  \Gamma\left(\frac{[M_R-M_T+2]-[2]+1}{2},\frac{2P}{\normSqr[\ba_i]}\right)\\
&\sim  \Gamma\left(\frac{\tilde{M}_R-\tilde{M}_T+1}{2},\frac{2P}{\normSqr[\tilde{\ba}_i]}\right)\\
&\sim LB(\tilde{\bH},\tilde{\ba}_i) 
\end{align*}\normalsize
\end{proof}

\Cref{claim:one_eq_behavior_depends_only_in_Mr_minus_Mt} shows that even if the channel is reduced to a $[2 \times (M_R-M_T+2)]$ channel size, the same lower bound can be used. However, this is not enough in order to claim that different blocks share the same lower bound, because the joint distribution of vectors from different blocks might be different. 
\begin{claim}[Reduction per block]\label{claim:channel_reduction}
Let \begin{align*}
\tilde{\ba}_1 &= [a_{11},a_{12}]^T \in \R^2 \\ 
\tilde{\ba}_2 &=[a_{21},a_{22}]^T \in \R^2\\
\ba_1 &= [0,\ldots 0, a_{11},a_{12}, 0,\ldots 0 ]^T \in \R^{M_T} \\ 
\ba_2 &=[0,\ldots 0, a_{21},a_{22},0,\ldots 0 ]^T \in \R^{M_T}\\
 \mathcal{A}&= [\ba_1,\ba_1] \in \R^{M_T \times 2}\\ 
\tilde{\mathcal{A}}&= [\tilde{\ba}_1,\tilde{\ba}_1] \in \R^{2 \times 2}. 
\end{align*}
Then, $LB(\bH,\mathcal{A}) \sim LB(\tilde{\bH},\tilde{\mathcal{A}})$,
where $LB(\bH,\mathcal{A}) = P{\left(\mathcal{A}^T{(\bH\bH^T)}^{-1}\mathcal{A}\right)}^{-1}$.
\end{claim}

\begin{proof}[Proof of \Cref{claim:channel_reduction}]
For convenience let us denote $\tilde{M}_R=M_R-M_T+2$, $\tilde{M}_T=2$, $rank(\mathcal{A})=R$. Then,
\small\begin{align*}
&LB(\bH,\mathcal{A}) \stackrel{(a)}{\sim} \\
&\mathcal{W}\left(P\left({\mathcal{A}^T\mathcal{A}}\right)^{-1},rank(\mathcal{A}),M_R-M_T+rank(\mathcal{A})\right)\\
&\sim\mathcal{W}\left(P\left({\mathcal{A}^T\mathcal{A}}\right)^{-1},R,M_R-M_T+ R\right)\\
&\sim\mathcal{W}\left(P\left({\mathcal{A}^T\mathcal{A}}\right)^{-1},R,[M_R-M_T+2] - [2] +R\right)\\
&\stackrel{(b)}{\sim}\mathcal{W}\left(P\left({\tilde{\mathcal{A}}^T\tilde{\mathcal{A}}}\right)^{-1},rank(\tilde{\mathcal{A}}),\tilde{M_R} - \tilde{M_T} +rank(\tilde{\mathcal{A}})\right)\\
&\stackrel{(a)}{\sim} LB(\tilde{\bH},\tilde{\mathcal{A}})
,
\end{align*}\normalsize
\end{proof}
where (a) follows from \cref{claim:dist_of_matrix_form_LB} and (b) is by \cref{eq:mathcalAA}.

\Cref{claim:channel_reduction,claim:one_eq_behavior_depends_only_in_Mr_minus_Mt} represent the distributions of lower bounds for values of $M_T$ and $M_R$. Nevertheless, the values $M_T$ and $M_R$ of do not matter. The only two things that do matter are the difference between $M_R$ and $M_T$ and whether of not the index of the vector $i$ is even. Hence, instead of analyzing the distribution of each tuple; $(i, M_R, M_T)$, it is enough to analyze the distribution of $(\tilde{i},\tilde{M}_R,\tilde{M}_T)$ where, $\tilde{i}=\frac{3+(-1)^i}{2}$, $\tilde{M}_R=2$ and $\tilde{M}_T=M_R-M_T$.  I.e., it is enough to know the index of the vector inside the block.

\begin{claim}[Two transmitters behave the same]\label{clm:first_and_sec_eq_set_behave_the_same}
Let $\tilde{M}_T=2$ and $\tilde{M}_R$ be the amount of receiving antennas. Then, the lower bound on the effective SNR achieved by the first and by the second vectors' sets are distributed the same.
\end{claim}
\begin{proof}[Proof of \Cref{clm:first_and_sec_eq_set_behave_the_same}]
For any full ranked constant $\bA\in \Z^{2 \times 2}$, by \Cref{claim:dist_of_matrix_form_LB},
\begin{align*}LB(\tilde{\bH},\bA)\sim  \mathcal{W}\left(P{(\bA^T\bA)}^{-1},2,\tilde{M}_R\right). \end{align*}
Thus, the distribution is dependent only in the quadratic form $\bA^T\bA$.
Notice that
\begin{align*}
\left(\bA_1^2\right)^T\bA_1^2 =  \begin{bmatrix}1&0\\0&1 \end{bmatrix} ;
\left(\bA_6^2\right)^T\bA_6^2= \begin{bmatrix}2&0\\0&2 \end{bmatrix}.  
\end{align*}
Hence, when the norm of both vectors is the same, whether it's 1, as in $\bA_1^2 =\begin{bmatrix}   1 & 0 \\   0 & 1 \end{bmatrix}$ or two, as in $\bA_6^2 =\begin{bmatrix}   1 & 1 \\   1 & -1\end{bmatrix}$, the vectors in $\bA$ are orthogonal and the effective SNR distributes the same for both diagonal entries.

Let us assume that the first vector is $[1,1]^T$, can we distinguishes between choosing the second vector as $[0,1]^T$ or as $[1,0]^T$? The selection of these vectors defines the matrices $\bA_2^{*2}=\begin{bmatrix}  1 & 1 \\	1 & 0	\end{bmatrix}$ and $\bA_4^{*2} =\begin{bmatrix}   1 & 0 \\   1 & 1\end{bmatrix}$, and the following quadratic forms:
\begin{align*}
\left(\bA_2^{*2}\right)^T\bA_2 =
 \begin{bmatrix}
  2 &1\\
  1 & 1
\end{bmatrix} ;
\left(\bA_4^{*2}\right)^T\bA_4  =  \begin{bmatrix}
  2 &1\\
  1 & 1
\end{bmatrix} 
\end{align*}
Thus, since the distribution is identical, we cannot distinguishes between the vectors. 

Finally, let us assume that the first vector is $[1,-1]^T$ and the second one is chosen between $[1,0]^T$  and $[0,1]^T$, as happens in $\bA_3^{*2}=\begin{bmatrix}  1 & 1\\ -1 & 0\end{bmatrix}$ and in $\bA_5^{*2} =\begin{bmatrix}   1 & 0 \\  -1 & 1\end{bmatrix}$ respectively. According to \cref{eq:IF_eq_snr_eff}, multiplying vectors by $-1$ does not change the rate of the effective SNR. Therefore, it is equivalent to choose the second vector between $[1,0]^T$ and $[0,-1]^T$ while the first vector remains the same. The quadratic form of both of them is 
\begin{align}
 \begin{bmatrix}
  1  & 1\\
  -1 & 0
\end{bmatrix} ^T
\begin{bmatrix}
  1  & 1\\
  -1 & 0
\end{bmatrix} &= 
\begin{bmatrix}
  2 & 1\\
  1 & 1
\end{bmatrix}
\\
 \begin{bmatrix}
  1  & 0\\
  -1 & -1
\end{bmatrix} ^T
\begin{bmatrix}
  1  & 0\\
  -1 & -1
\end{bmatrix} &= 
\begin{bmatrix}
  2 & 1\\
  1 & 1
\end{bmatrix}. 
\end{align}
Thus, since switching between vectors will end up at the exact same behavior, i.e., it will only switch the diagonal elements, both distributed the same. Concluding that the vectors with norm one can be switched without changing the distribution and the same is true for the vectors with norm two, which completes the proof.
\end{proof}

\begin{proof}[Proof of \Cref{lem:LB_for_each_eq_set_dist_identical}]
The proof is straightforward from \Cref{claim:one_eq_behavior_depends_only_in_Mr_minus_Mt,claim:channel_reduction,clm:first_and_sec_eq_set_behave_the_same}.
\end{proof}

\subsection*{Step 3 - Comparing between DSV and ZF}\label{subsubsec:step3}
In this step we want to make sure that even the lower bound on DSV is better, i.e., its CDF smaller then ZF's for any value of $t$. ZF has a gamma distribution, e.g., \cite[Lemma 1]{li2006distribution}. 

In DSV, from each set which contains two vectors, the one achieving the highest rate is selected.  According to \Cref{lem:LB_for_each_eq_set_dist_identical}, each channel can be reduced to a channel with two transmitting and $M_R-M_T+2$ receiving antennas without changing it's behavior. Moreover, both best vectors i.e., one from the first and one from the second set, behave the same. Thus, it is enough to understand the behavior of the best vector of the second set for two transmitting antennas. Let $\tilde{\mathcal{A}} \defined \mathcal{A}_2^2$ as was defined in \cref{eq:def_A_i} and let

\begin{align}
& m_{ij} \defined  [\tilde{\bS}^{-1}]_{ij} \text{ } \text{ } ,  \text{ }  m_i \defined m_{ii}.\label{eq:def_of_m}
\end{align}
Then,
\small\begin{align}
&DSV_{LB} \defined \max_{\ba \in\mathcal{A}_n^{M_T}}\left\{LB(\bH,\ba)\right\} \\
&\stackrel{(a)}{\sim} \max_{\tilde{\ba} \in\mathcal{A}_{\tilde{n}}^{2}}\left\{LB(\tilde{\bH},\tilde{\ba})\right\}\\
 &\sim P \cdot \max_{\tilde{\ba} \in\tilde{\mathcal{A}}}\left\{\left({  \tilde{\ba}^T\tilde{\bS}^{-1}\tilde{\ba} }\right)^{-1}\right\}\label{eq:why_con_b_is_important1}\\
 & \stackrel{(b)}{=}P\cdot \max\left\{ \frac{1}{ [\tilde{\bS}^{-1}]_{11} }, \frac{1}{[\tilde{\bS}^{-1}]_{11}+ [\tilde{\bS}^{-1}]_{22} - 2[\tilde{\bS}^{-1}]_{12}}\right\}\label{eq:why_con_b_is_important2}\\
&\stackrel{(c)}{=} P \cdot\max\left\{ \frac{1}{ m_{1} }, \frac{1}{m_{1}+ m_{2} - 2m_{12}}\right\}\label{eq:why_con_b_is_important3},
\end{align}\normalsize
where (a) is following from \Cref{lem:LB_for_each_eq_set_dist_identical}; (b) is by \cref{eq:def_a_i_1,eq:def_a_i_2,eq:def_A_i} and (c) is due to \cref{eq:def_of_m}.
Note that 
\begin{equation*}
 (\be_1)^T(\be_2-\be_2) = 1 \neq 0.
\end{equation*}
Therefore, by \cite[Theorem 5.3.1]{mathai1992quadratic}, the lower bound on the vectors in $\tilde{\mathcal{A}}$ are dependent.  Nonetheless, choosing the best one between them can only improve the effective SNR.

\begin{remark}
Since  $\tilde{\bS} \succ 0 \Longrightarrow \tilde{\bS}^{-1} \succ 0$, hence, $\forall i:m_i>0$. Thus, we can see that
\begin{equation}
m_{12} \leq 0 \Longrightarrow \max_{\tilde{\ba} \in\mathcal{A}_{\tilde{n}}^{2}}\left\{LB(\tilde{\bH},\tilde{\ba})\right\} = \frac{P}{ m_{1} }.
\end{equation}
\end{remark}

\begin{lemma}\label{lem:dif_ZF_LB_is_bounded_by_delta_lb}
Let
\small\begin{equation*}
\rho(\epsilon) =
\frac{4}{\pi} \int_0^{\frac{\pi}{4}} \left[F_\phi \left( 2(1+ \cos2 t) \cdot \left(\frac{\epsilon+1}{2\epsilon+1}\right)^2  \right) - F_{\phi}(1) \right]dt
\end{equation*}\normalsize
as was defined in \Cref{thm:CDF_of_LB_is_bounded_by_time_sharing}.
Then, for any $\epsilon \in \left(0,\frac{1}{\sqrt{2}} \right)$,
\begin{equation}
F_{DSV_{LB}}(t) 
\leq  F_X(t)  - \rho(\epsilon)\left[ F_{X\vert a_\epsilon}(t)-F_{X\vert a_\epsilon}\left(\frac{t}{1+\epsilon}\right) \right],
\end{equation}
where $\theta \sim F(K,K)$ and $F_{DSV_{LB}}(t)$ stands for the CDF of 
\begin{equation*} 
 \max_{\tilde{\ba} \in\mathcal{A}}\left\{LB(\tilde{\bH},\tilde{\ba})\right\}
\end{equation*}
\end{lemma}
In order to prove \Cref{lem:dif_ZF_LB_is_bounded_by_delta_lb}, \Cref{clm:bound_dif_zf_and_lbM,clm:lower_bound_a_epsilon_by_rho} are required.

\begin{remark}
\begin{equation*}F_\phi(t) = I_{\frac{K \cdot t}{K \cdot t + K}}\left (\tfrac{K}{2}, \tfrac{K}{2} \right),\end{equation*}
where $I$ is the regularized incomplete beta function.
Since $\forall t \in (0,\pi/4)$ and $\forall \epsilon \in \left(0,\frac{1}{\sqrt{2}}\right)$, $ 2 \left(\frac{\epsilon+1}{2\epsilon+1}\right)^2>1$ and $\cos 2t \geq 0$, $\rho(\epsilon)$ is always positive.
\end{remark}
\begin{example}\label{ex:rho_for_K_equals_2}
Note that if $M_R=M_T$, then $K=2$. Hence,\begin{align*}
F_\theta(t)= I_{\frac{t}{t+1}(1,1)} = \frac{t}{t+1}\cdot u(t) ,
\end{align*}
where $u(t)$ is the Heaviside step function. In that case, we can use Wolfram Mathematica 10.4, a computation program and find $\rho(\epsilon)$:
\begin{equation*}
\rho(\epsilon) = \frac{1}{2}- \frac{4(1+2\epsilon)}{\pi\sqrt{5+4\epsilon(3+2\epsilon)}}\arctan\left( \frac{1+2\epsilon}{\sqrt{5+4\epsilon(3+2\epsilon)}} \right)
\end{equation*}
\begin{equation*}
\rho(0) = \frac{1}{2} -  \frac{4\arctan\left(\frac{1}{\sqrt{5}}\right)}{\sqrt{5}\pi} \approx 0.26
\end{equation*}
and 
\begin{align*}
\rho\left(\frac{1}{\sqrt{2}}\right) &= \frac{1}{2} - \frac{4(1+\sqrt{2})\arctan\left(\frac{4(1+\sqrt{2})}{9+6\sqrt{2}}\right)}{\sqrt{9+6\sqrt{2}}}\\
&= \frac{1}{2} - \frac{4}{\pi \sqrt{3}}\arctan\left( \frac{1}{\sqrt{3}}  \right)\\
& = \frac{1}{2}-\frac{2}{3\sqrt{3}}\\
& \approx 0.115
\end{align*}
Since $\left(\frac{\epsilon+1}{2\epsilon+1}\right)^2$ monotonously decreases as a function of $\epsilon$,  the same is true for $\rho(\epsilon)$. Concluding that when $M_R=M_T$, $0.115 \leq \rho(\epsilon) \leq 0.26$. Recalling the $\rho(\epsilon)$ is a lower bound on the 
probability that the the effective SNR of the second vector $\be_1 - \be_2$  is better then  the one achieved by the unit vector $\be_1$.

\end{example}

In order to prove \Cref{lem:dif_ZF_LB_is_bounded_by_delta_lb}, $F_{DSV_{LB}}(t)$, the CDF of 
 $\max_{\ba \in\mathcal{A}_n^{M_T}}\left\{LB(\bH,\ba)\right\}$
 is upper bounded by \Cref{clm:bound_dif_zf_and_lbM}.
\begin{claim}\label{clm:bound_dif_zf_and_lbM}
Let
\begin{align*}
& X = LB(\bH,\be_1); & Y =  LB(\bH,\be_1-\be_2) &;& \bS = \bH^T\bH.
\end{align*}
Let us define the condition $a_\epsilon$  by $a_{\epsilon} : Y > X(1+ \epsilon)$. Then,
\small\begin{align}
F_{DSV_{LB}}(t) &\leq 
 (1-\Pr\{a_\epsilon\}) \cdot F_{X\vert a_\epsilon}(t) + \Pr\{a_\epsilon\} \cdot F_{X\vert a_\epsilon}\left(\frac{t}{1+\epsilon}\right), 
\end{align}\normalsize
where $F_{X\vert a_\epsilon}(t)$ stands for the CDF of $X$ given $a_\epsilon$.
\end{claim}
\begin{proof}[Proof of \Cref{clm:bound_dif_zf_and_lbM}]
For convenience, let us denote the reduction of $\bS$ according to \Cref{lem:LB_for_each_eq_set_dist_identical} by $\tilde{\bS}$, and let
\begin{align}
m_{ij} &= [\tilde{\bS}^{-1}]_{ij},\\
m_{i} &= [\tilde{\bS}^{-1}]_{ii}.
\end{align}
The condition  $b$ is defined as 
\begin{equation}
b :  m_{12}>0 	\stackrel{ \text{equivalent }}{\Longleftrightarrow} \cos\theta>0,\label{eq:condition_b}
\end{equation}
where $\theta$ is the angle between $m_1$ and $m_2$. Note that $m_1,m_2 >0$ since $\bS$ is a positive definite matrix almost surly, hence,
\begin{equation}
a_\epsilon \Longrightarrow b.
\label{eq:q_epsilon_causes_b}
\end{equation} 
\Cref{eq:why_con_b_is_important1,eq:why_con_b_is_important2,eq:why_con_b_is_important3} are the explanation for condition $b$.
\begin{align*}
&F_X(t) \\
&= \Pr\{X \leq t\} = \underbrace{\Pr\left\{\bar{b}\right\}\Pr\left\{ X \leq t \big{\vert} \bar{b}\right\}}_{\gamma_1} + \p[b]\Pr\left\{ X \leq t \big{\vert} b \right\}\\
&= (\gamma_1) +  \Pr\left\{b\right\}\Pr\left\{\bar{a_{\epsilon}} \big{\vert} b \right\}   \Pr\left\{ X \leq t \big{\vert} \bar{a_{\epsilon}},b \right\}\\
&  +  \Pr\left\{b\right\}\Pr\left\{a_{\epsilon} \big{\vert} b \right\}   \Pr\left\{ X \leq t \big{\vert} a_{\epsilon}, b \right\}  \\
&= (\gamma_1) +  \underbrace{\Pr\left\{\bar{a_{\epsilon}} , b \right\}   \Pr\left\{ X \leq t \big{\vert} \bar{a_{\epsilon}},b \right\}}_{\gamma_2}\\
&  + \Pr\left\{a_{\epsilon} , b \right\}   \Pr\left\{ X \leq t \big{\vert} a_{\epsilon}, b \right\}\\
&= (\gamma_1) +  (\gamma_2) + \Pr\left\{a_{\epsilon},b\right\}\Pr\left\{ X \leq t \big{\vert} a_{\epsilon}, b \right\}
.
\end{align*}
Next, we do the same to $F_{DSV_{LB}}(t)$:
\begin{align*}
&F_{DSV_{LB}}(t)  = \Pr\{\max(X,Y) \leq t\}\\
& = \Pr\left\{\bar{b}\right\}\Pr\left\{ \max(X,Y) \leq t \big{\vert} \bar{b}\right\}\\
& + \Pr\left\{b\right\}\Pr\left\{ \max(X,Y) \leq t \big{\vert} b \right\}\\
&\stackrel{(a)}{=} \underbrace{\Pr\left\{\bar{b}\right\}\Pr\left\{ X \leq t \big{\vert} \bar{b}\right\}}_{\gamma_1} + \Pr\left\{b\right\}\Pr\left\{ \max(X,Y) \leq t \big{\vert} b \right\}\\
&= (\gamma_1) + \p[b] \p[\bar{a_{\epsilon}}\big{\vert} b]\p[\max(X,Y) \leq t \big{\vert} \bar{a_{\epsilon}}, b]   \\
&+\p[b]\p[a_{\epsilon}\big{\vert}b] \p[\max(X,Y) \leq t \big{\vert} a_{\epsilon}, b] \\
&= (\gamma_1) 
+  \p[\bar{a_{\epsilon}}, b] \underbrace{\Pr\left\{ \max(X,Y) \leq t \big{\vert} \bar{a_{\epsilon}}, b \right\}}_{\bar{a_{\epsilon}} \Longrightarrow X \geq Y} \\ 
&+ \Pr\{a_{\epsilon} , b \} \underbrace{\Pr\left\{ \max(X,Y) \leq t \big{\vert} a_{\epsilon}, b \right\} }_{a_{\epsilon} \Longrightarrow X < Y} \\
&= (\gamma_1) 
+  \underbrace{\Pr\{\bar{a_{\epsilon}}, b \} \Pr\left\{ X \leq t \big{\vert} \bar{a_{\epsilon}}, b \right\}}_{\gamma_2}  
+ \Pr\{a_{\epsilon} , b \} \Pr\left\{ Y \leq t \big{\vert} a_{\epsilon}, b \right\} \\
&= (\gamma_1) + (\gamma_2) +   \Pr\{a_{\epsilon}, b \} \Pr\left\{ Y \leq t \big{\vert} a_{\epsilon}, b \right\}, 
\end{align*}
where (a) is by \cref{eq:q_epsilon_causes_b}.
Concluding that, 
\begin{align*}
&F_X(t)-F_{DSV_{LB}}(t) \\
&=    \Pr\left\{a_{\epsilon},b\right\}\Pr\left\{ X \leq t \big{\vert} a_{\epsilon}, b \right\}-\Pr\{a_{\epsilon}, b \} \Pr\left\{ Y \leq t \big{\vert} a_{\epsilon}, b \right\}\\
&=
\Pr\left\{a_{\epsilon},b\right\}\left[{  \Pr\left\{ X \leq t \big{\vert} a_{\epsilon},b \right\} 
-
 \Pr\left\{ Y \leq t \big{\vert} a_{\epsilon},b \right\}  }\right]\\
&=
 \p[a_{\epsilon},b] \cdot\\
&\left[{  \p[ X \leq t \big{\vert} Y > X(1+ \epsilon), b ]
-
 \p[ Y \leq t \big{\vert} Y > X(1+ \epsilon), b]  }\right]\\
&\stackrel{(b)}{\geq}\p[a_{\epsilon},b]\p[ X \leq t \big{\vert} Y > X(1+ \epsilon), b] \\
&-\p[a_{\epsilon},b]\p[ X \leq \frac{t}{1+\epsilon} \big{\vert} Y > X(1+ \epsilon), b]
&\\
&= \p[a_{\epsilon},b]\int_{\frac{t}{1+\epsilon}}^{t}f_{X\big{\vert} a_{\epsilon},b }(x) dx\\
&\stackrel{(c)}{=} \p[a_{\epsilon}]\int_{\frac{t}{1+\epsilon}}^{t}f_{X\big{\vert} a_{\epsilon} }(x) dx\\
&\geq
\p[a_{\epsilon}]\left[ F_{X\vert a_\epsilon}(t)-F_{X \vert a_\epsilon}\left(\frac{t}{1+\epsilon}\right) \right],
\end{align*} 
where (b) comes from the def. on $a_\epsilon$ in \Cref{thm:CDF_of_LB_is_bounded_by_time_sharing}  and (c) is due to  \cref{eq:q_epsilon_causes_b}.
\end{proof}

\begin{claim}\label{clm:lower_bound_a_epsilon_by_rho}
For any $\epsilon \in \left(0,\frac{1}{\sqrt{2}}  \right)$,
$\Pr\{a_{\epsilon}\}\geq \rho(\epsilon)$. 
Recalling that 
\begin{equation*}
\rho(\epsilon) =
\frac{4}{\pi} \int_0^{\frac{\pi}{4}} \left[F_\phi \left( 2(1+ \cos2 t) \cdot \left(\frac{\epsilon+1}{2\epsilon+1}\right)^2  \right) - F_{\phi}(1) \right]dt.
\end{equation*}
as was defined in \Cref{thm:CDF_of_LB_is_bounded_by_time_sharing}.
\end{claim}
\begin{example}
When $M_T=M_R$, we can calculate numerically both $\rho(\epsilon)$ and $\p[a_\epsilon]$. For example,
\begin{align*}
&\p[a_0] \approx 0.276 > \rho(0) \approx 0.26\\ 
&\p[a_{\left(1/\sqrt{2}	\right)}] \approx 0.161 > \rho\left(\frac{1}{\sqrt{2}}\right) \approx 0.115.
\end{align*}
\end{example}
\begin{proof}[Proof of \Cref{clm:lower_bound_a_epsilon_by_rho}]
For convenience, let us denote $\deEp = \frac{\epsilon}{1+\epsilon}$. We use same notations for
$\bS,m_{ij},m_i,a_{\epsilon},b$  as were used in \Cref{clm:bound_dif_zf_and_lbM}.

\begin{align*}
\p[a_\epsilon] 
&\stackrel{(a)}{=} 
\p[a_{\epsilon},b] = \p[ Y > X(1+\epsilon), b] \\
&= \p[LB(\bH,\be_1-\be_2) > (1+\epsilon)LB(\bH,\be_1)]\\
 &\stackrel{(b)}{=} 
\p[LB(\tilde{\bH},\tilde{\be}_1-\tilde{\be}_2) > (1+\epsilon)LB(\tilde{\bH},\tilde{\be}_1)]\\
&= \p[\frac{P}{m_1-2m_{12}+m_2}> \frac{P(1+\epsilon)}{m_1}]\\
&= \p[m_1 > (1+\epsilon)(m_1-2m_{12}+m_2)]\\
&=\p[-2(1+\epsilon)m_{12} > (1+\epsilon)m_2 + \epsilon m_1]     \\
&= \p[-2m_{12} > m_2 + \delta_\epsilon m_1], 
\end{align*}
where (a) follows from \cref{eq:q_epsilon_causes_b} and (b) is according to \Cref{lem:LB_for_each_eq_set_dist_identical}.
\begin{align*}
&\stackrel{(c)}{\geq} 
\p[b]\p[{\deEp}  m_1 +m_2 < -2m_{12} \big{\vert} b]\\
&= \p[b,m_1>m_2]\p[{\deEp}  m_1 +m_2 < -2m_{12} \big{\vert} b , m_1>m_2]\\ 
&  +\p[b,m_1<m_2]\p[{\deEp}  m_1 +m_2 < -2m_{12} \big{\vert} b , m_1<m_2]\\
&\stackrel{(d)}{\geq }
\p[b,m_1>m_2]\p[{\deEp} m_1 +m_1 < -2m_{12} \big{\vert} b , m_1>m_2]\\
& +\p[b,m_1<m_2]\p[{\deEp} m_2 +m_2 < -2m_{12} \big{\vert} b , m_1<m_2]\\
&=
\p[b,m_1>m_2]\p[(\deEp+1)  m_1 < -2m_{12} \big{\vert} b , m_1>m_2]\\
&  +\p[b,m_1<m_2]\p[(\deEp+1)  m_2 < -2m_{12} \big{\vert} b , m_1<m_2]\\
&\stackrel{(e)}{=}
\p[b]\p[(\deEp+1)  m_1 < -2m_{12} \big{\vert} b,m_1>m_2]\\
&\stackrel{(f)}{=} 
\p[m_{12}<0]\cdot \\
&\cdot \p[(\deEp+1)  m_1 < -2m_{12} \big{\vert} m_{12}<0,m_1>m_2],
  \end{align*}
where (c) is due to law of total probability;
(d) is because $\{m_i \deEp +m_j \big{\vert} m_i \leq m_j \} \leq \{m_i \deEp +m_i \big{\vert} m_i \leq m_j \} = \{m_i (\deEp +1  )\big{\vert} m_i \leq m_j \}$; (e) is following from the fact that $m_1$ and $m_2$ are distributed the same and (f) is by \Cref{eq:condition_b}.
\small\begin{align*}
&=\p[{[\tilde{\bS}^{-1}]_{12}}<0] \cdot \\
&\p[{
(\deEp+1) [\tilde{\bS}^{-1}]_{11} < -2[\tilde{\bS}^{-1}]_{12} \big{\vert} [\tilde{\bS}^{-1}]_{12}<0,  [\tilde{\bS}^{-1}]_{11}>  [\tilde{\bS}^{-1}]_{22}
}]\\
&\stackrel{(g)}{=}
\p[\frac{-[\tilde{\bS}]_{21}}{\det(\bS)}<0] \cdot\\
&\p[(\deEp+1)  \frac{[\tilde{\bS}]_{22}}{\det(\bS)} < -2\frac{-[\tilde{\bS}]_{21}}{\det(\bS)} \big{\vert} \frac{-[\tilde{\bS}]_{21}}{\det(\bS)}<0, \frac{[\tilde{\bS}]_{22}}{\det(\bS)}> \frac{[\tilde{\bS}]_{11}}{\det(\bS)}]\\
&\stackrel{(h)}{=} 
\p[{[\tilde{\bS}]_{21}>0}] \cdot \\
&\p[{
(\deEp+1)[\tilde{\bS}]_{22} <2[\tilde{\bS}]_{21} \big{\vert} [\tilde{\bS}]_{21}>0,[\tilde{\bS}]_{22}>[\tilde{\bS}]_{11}
}]\\
&\stackrel{(\gamma_1)}{=}  \p[\langle \tilde{\bh}_1 , \tilde{\bh}_2\rangle>0] \cdot \\
&\p[
(\deEp+1)\langle \tilde{\bh}_2 , \tilde{\bh}_2\rangle <2\langle \tilde{\bh}_1 , \tilde{\bh}_2\rangle \big{\vert} \langle \tilde{\bh}_1 , \tilde{\bh}_2\rangle>0, \langle \tilde{\bh}_2 , \tilde{\bh}_2\rangle > \langle \tilde{\bh}_1 , \tilde{\bh}_1\rangle
]\\
&\stackrel{(j)}{=}
\p[\cos \theta>0] \cdot \\
&\p[{
(\deEp+1)\normSqr[\tilde{\bh_2}] < 2\|\tilde{\bh_1}\|\|\tilde{\bh_2}\|\cos \theta \big{\vert} \cos \theta>0, \normSqr[\tilde{\bh_2}]>\normSqr[\tilde{\bh_1}]
}]\\
&= \p[\cos \theta>0] \cdot \\
&\p[{
(\deEp+1) \frac{\|\tilde{\bh_2}\|}{\|\tilde{\bh_1}\|} <2\cos \theta
 \big{\vert} \cos \theta>0, \normSqr[\tilde{\bh_2}]>\normSqr[\tilde{\bh_1}]
}]\\
&=  \p[\cos \theta>0]\p[{
 \frac{\|\tilde{\bh_2}\|}{\|\tilde{\bh_1}\|} <\frac{2\cos \theta}{\deEp+1}
 \big{\vert} \cos \theta>0, \frac{\normSqr[\tilde{\bh_2}]}{\normSqr[\tilde{\bh_1}]}>1
}],
\end{align*}
where (g) is because $\tilde{\bS} \in \R^{2 \times 2}$ ; (h) is since $\tilde{\bS}$ is positive definite matrix almost surly, hence $\det(\tilde{\bS})>0$; $(\gamma_1)$ is following from the definition of $\tilde{\bS} = \tilde{\bH}^T\tilde{\bH}$ where $\tilde{\bh}_i$ is the $i^{th}$ column of $\tilde{\bH}$ and (j) is by the definition of inner product where $\theta$ is the angle between the vectors $\tilde{\bh}_1$ and $\tilde{\bh}_2$.
\begin{align*}
&=  \p[\cos \theta>0]\p[{
 1 < \frac{\|\tilde{\bh_2}\|}{\|\tilde{\bh_1}\|} <\frac{2\cos \theta}{\deEp+1}
 \big{\vert} \cos \theta>0, \frac{\normSqr[\tilde{\bh_2}]}{\normSqr[\tilde{\bh_1}]}>1
}]\\
&\stackrel{(k)}{=}  \p[{
 1 < \frac{\|\tilde{\bh_2}\|}{\|\tilde{\bh_1}\|} <\frac{2\cos \theta}{\deEp+1}
 , \cos \theta>0 \big{\vert}  \frac{\normSqr[\tilde{\bh_2}]}{\normSqr[\tilde{\bh_1}]}>1
}]\\
&=\p[{
 1 < \frac{\|\tilde{\bh_2}\|}{\|\tilde{\bh_1}\|} <\frac{2\cos \theta}{\deEp+1}
 \big{\vert}  \frac{\normSqr[\tilde{\bh_2}]}{\normSqr[\tilde{\bh_1}]}>1
}]\\
&\stackrel{(k)}{=} \frac{1}{\p[{\frac{\normSqr[\tilde{\bh_2}]}{\normSqr[\tilde{\bh_1}]}>1}]}\p[{
 1 < \frac{\|\tilde{\bh_2}\|}{\|\tilde{\bh_1}\|} <\frac{2\cos \theta}{\deEp+1}
,  \frac{\normSqr[\tilde{\bh_2}]}{\normSqr[\tilde{\bh_1}]}>1
}]\\
&\stackrel{(l)}{=} 2\p[{
 1 < \frac{\|\tilde{\bh_2}\|}{\|\tilde{\bh_1}\|} <\frac{2\cos \theta}{\deEp+1}
}]\\
&= 2 \p[{
 1<\frac{\normSqr[\tilde{\bh_2}]}{\normSqr[\tilde{\bh_1}]} <\left(\frac{2}{\deEp+1}\right)^2\cos^2 \theta
}]\\
&\stackrel{(m)}{=} 2 \p[{
 1<\phi <\left(\frac{2}{\deEp+1}\right)^2\cos^2 \theta
}],
\end{align*}
where (k) is following from Bayes' theorem; (l) is because  $\tilde{\bh}_1,\tilde{\bh}_2$ are i.i.d; (m) by definition $\phi \sim F(M_R-M_T+2,M_R-M_T+2)$. 

Note that
$\forall \epsilon \in \left(0,\frac{1}{\sqrt{2}}\right): \frac{2}{(\deEp+1)^2} = \frac{2}{(\frac{\epsilon}{1+\epsilon}+1)^2} > 1$.
Hence, as long as $\cos 2\theta > 0$,
\begin{align*}
\frac{2}{(\deEp+1)^2} +  \frac{2\cos2\theta}{(\deEp+1)^2}>1.
\end{align*}

Thus,
\small\begin{align*}
 &2 \p[{
 1<\phi <\left(\frac{2}{\deEp+1}\right)^2\cos^2 \theta
}]\\
&= 2 \p[{
 1<\phi <\frac{2}{(\deEp+1)^2}\cdot2\cos^2 \theta
}]\\
&= 2 \p[{
 1<\phi <\frac{2(1+ \cos2\theta)}{(\deEp+1)^2}
}]\\
&= 2 \p[{
 1<\phi <\frac{2}{(\deEp+1)^2} +  \frac{2\cos2\theta}{(\deEp+1)^2}
}]\\
&= 2 \p[{
 1<\phi <\frac{2}{(\deEp+1)^2} +  \frac{2\cos2\theta}{(\deEp+1)^2}
}]\\ 
&\geq 2\int_{\cos \theta >0} \left[F_\phi \left( \frac{2}{(\deEp+1)^2} +  \frac{2\cos2 t}{(\deEp+1)^2} \right) - F_{\phi}(1) \right]F_\theta(t)dt\\
&= 2\cdot \frac{1}{2\pi}\int_{\cos \theta >0} \left[F_\phi \left( \frac{2}{(\deEp+1)^2} +  \frac{2\cos2 t}{(\deEp+1)^2} \right) - F_{\phi}(1) \right]dt\\
&= 2\cdot \frac{1}{2\pi}\cdot 4 \int_0^{\frac{\pi}{4}} \left[F_\phi \left( \frac{2}{(\deEp+1)^2} +  \frac{2\cos2 t}{(\deEp+1)^2} \right) - F_{\phi}(1) \right]dt\\
&=  \frac{4}{\pi} \int_0^{\frac{\pi}{4}} \left[F_\phi \left( \frac{2}{(\deEp+1)^2} +  \frac{2\cos2 t}{(\deEp+1)^2} \right) - F_{\phi}(1) \right]dt\\
&=  \frac{4}{\pi} \int_0^{\frac{\pi}{4}} \left[F_\phi \left( \frac{2(1+\cos 2t)}{(\deEp+1)^2}  \right) - F_{\phi}(1) \right]dt\\
&=  \frac{4}{\pi} \int_0^{\frac{\pi}{4}} \left[F_\phi \left( 2(1+\cos 2t) \left(\frac{\epsilon+1}{2\epsilon+1}\right)^2  \right) - F_{\phi}(1) \right]dt\\
&\stackrel{(a)}{=} \rho(\epsilon),
\end{align*}
where (a) is following from the definition of $\rho(\epsilon)$ in \Cref{thm:CDF_of_LB_is_bounded_by_time_sharing}.
\end{proof}

\begin{proof}[Proof of \Cref{lem:dif_ZF_LB_is_bounded_by_delta_lb}]
The proof is straightforward from \Cref{clm:bound_dif_zf_and_lbM,clm:lower_bound_a_epsilon_by_rho}.
\end{proof}

\begin{proof}[Proof of \Cref{thm:CDF_of_LB_is_bounded_by_time_sharing}]
The proof is straightforward from \Cref{lem:dif_ZF_LB_is_bounded_by_delta_lb} and by recalling that DSV is upper bounded by NB-IF, hence, the value of its CDF is never smaller than the CDF of NB-IF
\end{proof}